\documentclass[authoryear,a4paper, 12pt]{elsarticle}
\usepackage{natbib}
\usepackage{amsthm}
\usepackage{amssymb}
\usepackage{graphicx,amsmath}
\usepackage{tikz}
\usepackage{mathtools}
\usepackage{natbib}
\usepackage{booktabs,xcolor}
\usepackage{graphics}
\usepackage{enumitem}
\usepackage{subcaption}
\usepackage{appendix}
\usepackage{adjustbox}
\usepackage{textcomp}
\usepackage{caption}
\usepackage{setspace}
\usepackage[colorlinks=true,linkcolor=blue, urlcolor=blue, citecolor=blue, breaklinks]{hyperref}
\usepackage{pifont}
\journal{Games and Economic Behavior}

\newtheorem{proposition}{Proposition}

\newtheorem{lemma}{Lemma}
\newtheorem{result}{Result}
\theoremstyle{definition}

\newtheorem{hypothesis}{Hypothesis}
\newtheorem{definition}{Definition}

\newtheorem*{lemma*}{Lemma}
\newtheorem*{proposition*}{Proposition}

\usepackage{etoolbox}
\usepackage[flushleft]{threeparttable}
\makeatletter
\patchcmd{\ps@pprintTitle}
{Preprint submitted to}
{Accepted to }
{}{}
\makeatother
\patchcmd{\emailauthor}{(#2)}{}{}{}
\patchcmd{\urlauthor}{(#2)}{}{}{}

\newtheorem*{theorem*}{Theorem}
\theoremstyle{plain} 
\newcommand{\thistheoremname}{}
\newtheorem*{genericthm*}{\thistheoremname}
\newenvironment{namedthm*}[1]
{\renewcommand{\thistheoremname}{#1}%
\begin{genericthm*}}
{\end{genericthm*}}

\definecolor{ForestGreen}{rgb}{.13,.54,.13}

\begin{document}
\begin{frontmatter}

\title{Fair Cake-Cutting in Practice\footnote{We acknowledge excellent research assistance by Andr\'es Azqueta-Galvad\'on, Wai Chung, Mingxian Jin, Alexander Sauer and Carlo Stern, and helpful comments from Uriel Feige,
Jeph Herrin, Herv\'e Moulin and the anonymous referees of this journal and the 20th ACM conference on Economics and Computation, where a preliminary version of this paper has been presented as an extended abstract.
~\\
We thank Sarah Fox and Erin Goldfinch for proof-reading the paper. This project was funded by the EssexLab at and the Centre for European Economic Research, and received ethical approval from the University of Essex. Our data can be downloaded from the \href{http://researchdata.essex.ac.uk/83/}{University of Essex Data Repository}. Replication materials are available at \href{www.josueortega.com}{www.josueortega.com}. We are grateful to COST Action IC1205 for organizing the FairDiv-15 summer school that introduced us to cake-cutting. Erel is partly funded by the Israel Science Foundation grant 712/20.}}
\author[add1]{Maria Kyropoulou}

\author[add2]{Josu\'e Ortega}
		\ead{j.ortega@qub.ac.uk (Corresponding author).}
\author[add4]{Erel Segal-Halevi}

\address[add1]{School of Computer Science and Electronic Engineering, University of Essex, UK.}
\address[add2]{Queen's Management School, Queen's University Belfast, UK.}
\address[add4]{Department of Computer Science, Ariel University, Israel.}	

\date{26 January 2022}

\begin{abstract}
Using two lab experiments, we investigate the real-life performance of envy-free and proportional cake-cutting procedures with respect to fairness and preference manipulation. 
Although the observed subjects' strategic behavior eliminates the fairness guarantees of envy-free procedures, we nonetheless find evidence that suggests that envy-free procedures are fairer than their proportional counterparts. 

Our results support the practical use of the celebrated Selfridge-Conway procedure, and more generally, of envy-free cake-cutting mechanisms.
We also find that subjects learn their opponents' preferences after repeated interaction and use this knowledge to improve their allocated share of the cake. Learning increases strategic behavior, but also reduces envy.
\end{abstract}

\end{frontmatter}

\newpage
\setcounter{footnote}{0}
\section{Introduction}
\label{sec:introduction}
The problem of how to fairly divide a heterogeneous good among agents who value different parts of it distinctly has been thoroughly studied in many areas of science over the last seventy years. The heterogeneous good is often referred to as the cake, and thus this problem is known as cake-cutting (see \citealp{brams1996}, \citealp{robertson1998} and \citealp{moulin2004} for textbook references). 
Although fundamental breakthroughs have been achieved on the construction of fair cake-cutting procedures, the question of how these procedures fare when applied to real people has not yet been tackled. 
This paper reports the results of two laboratory experiments that provide insights on this question.

Let us start by clarifying what we mean by fair. Although several notions of fairness have been proposed, two important ones stand out for their intuitive formulation. The first one is {\it proportionality}, which requires that every agent obtains at least what she considers to be $1/n$ of the cake when dividing a cake among $n$ agents \citep{steinhaus1948}. The second one is {\it envy-freeness}, which demands that no agent prefers the share allocated to any other agent over hers \citep{gamow1958, foley1967}. Any division that is envy-free is also proportional (if the entire cake is allocated) but the converse is not true, and thus envy-freeness is a stronger property than proportionality.\footnote{Envy-freeness and proportionality are equivalent in the two-agent case.} Proportionality and envy-freeness are often considered ``{\it the two most important tests of equity}'' \cite[p. 166]{moulin2014}.

The literature has developed procedures that produce envy-free cake divisions when all agents report their preferences over the cake pieces non-strategically; we will refer to these as envy-free procedures. 
However, if agents strategically misrepresent their preferences, an allocation with envy can be obtained as a Nash equilibrium outcome of the game associated to envy-free procedures \citep{branzei2016}. In fact, as we show in Proposition \ref{lem:envy-when-strategizing}, envy can rationally emerge in envy-free procedures even when only one agent behaves strategically. Therefore, a key question is whether envy-free cake-cutting procedures are manipulated in practice, and whether such manipulations, if they exist, lead to envy. This is the first question that we tackle in this paper. If envy-free procedures generate envy because of strategic behavior, there would be little support for their real-life implementation, in particular because envy-free mechanisms are particularly involved: the Selfridge-Conway procedure for three agents \citep{brams1996}, studied in the present paper, is a case in point.

A second related question that we tackle is whether agents can successfully learn their opponents' preferences through repeated interaction. This question relates to the previous one in that an agent needs some information about their opponents' preferences to successfully manipulate a cake-cutting procedure. An agent can acquire this valuable information through experimentation, i.e. varying her strategies over time and observing her opponents' best responses to them. If agents do not learn through experimentation, there is little concern about the manipulation of cake-cutting procedures in environments in which agents' preferences are privately known, and thus no concern about the emergence of envy in otherwise envy-free procedures.

We tackle these two questions by means of two lab experiments. We study:
\begin{enumerate}
\item the fairness (i.e. envy-freeness) of envy-free and proportional cake-cutting procedures,

\item the extent to which agents manipulate those procedures, and

\item whether agents learn their opponents' preferences and use that information to their advantage.
\end{enumerate}

We consider the most popular cake-cutting procedures and compare their theoretical properties against their real performance in the lab. The procedures we consider are:
\begin{itemize}
\item For 2 agents: symmetric and asymmetric cut-and-choose, and (a discrete adaptation of) Dubins-Spanier moving knife;
\item For 3 agents: Knaster-Banach last diminisher, (a discrete adaptation of) Dubins-Spanier moving knife, and Selfridge-Conway;
\item For 4 agents: Knaster-Banach last diminisher, (a discrete adaptation of) Dubins-Spanier moving knife, and Even-Paz.
\end{itemize}

These cake-cutting procedures, described in detail in the next section, are well-known in the literature because they all achieve proportional allocations. Furthermore, the asymmetric and symmetric cut-and-choose and the Selfridge-Conway procedures are even envy-free.\footnote{We do not include an envy-free procedure for four agents because the only finite ones known to date \citep{Aziz2015Discrete,amanatidis2018improved} require over 100 queries.}

In our two experiments (henceforth EXP1 and EXP2), 133 and 114 subjects divide cakes in several rounds using the aforementioned procedures. The main difference between EXP1 and EXP2 is that in EXP1, subjects divide cakes against automata, who are programmed to act non-strategically, whereas in EXP2 agents divide cakes against real participants, who may act strategically. 

Each cake is divided 7 times in what we call rounds, during which their opponents' preferences remain constant. This gives subjects the opportunity to try to learn their opponents' preferences. The reward structure of the experiment is such that subjects are actually incentivized to learn, as they get paid an amount that depends on how much they value their allocated share, i.e. their payoff, at each round. In addition, during the final two rounds agents are directly told their opponents' preferences, so that we are able to differentiate between manipulations made to learn the opponents' preferences and those made to directly increase the subjects' immediate payoff. Subjects observe which share of the cake they get in each round and the value (in their own eyes) of their opponents' shares.

\subsection{Overview of Results}
We find that all cake-cutting procedures are very frequently manipulated (in up to 85\% of the cases for some mechanisms, and at least in 40\% of the cases for all mechanisms). The only mechanism in which non-strategic behavior is consistently more frequently observed than manipulations is Selfridge--Conway (subsection \ref{sub:results-manipulation}). As a consequence, envy-free procedures generate envy. Envy is generated in 4--7\% of cases when using the asymmetric cut-and-choose procedure in which the subject cuts the cake, in 18--24\% of cases when using the symmetric cut-and-choose procedure in which both subjects cut the cake, and in 16--28\% of cases when using Selfridge-Conway. However, these procedures still generate significantly less envy than their proportional counterparts (subsection \ref{sub:results-envy}).

Overall, the experimental results provide support for the use of the cut-and-choose and Selfridge-Conway procedures, and more generally, of envy-free cake-cutting procedures. These procedures are less manipulated in practice and generate substantially less envy than proportional ones.

We find some evidence of successful learning, in particular in the cut-and-choose procedure, the Knaster-Banach last diminisher, and to some extent in Selfridge-Conway. Surprisingly, we observe that more knowledge does not always yield higher payoffs. This is because agents use that knowledge to manipulate the cake-cutting procedures in the wrong way. In particular, they try to follow simple heuristics that worked in the past, such as cutting the cake {\it a bit more to the right}, which may be harmful in other procedures in which the optimal manipulation was to cut the cake {\it a bit more to the left}. Overall, we observe that knowledge significantly decreases the level of non-strategic behavior and envy (subsection \ref{sub:results-learning}).

Moreover, we find in two-agent procedures that between 70\% to 76\% of the agents do manipulations that are clearly harmful to them, even in the simple cut-and-choose procedure. For example,
they cut the cake at a certain location $x$, see that their partner chooses the right piece, and then, at the next play against the same partner, cut to the left of $x$ --- which is guaranteed to result in a smaller piece for them (subsection \ref{sub:results-quality}).

\paragraph{\textbf{Structure of the article}}
Section \ref{sec:literature} presents an overview of related experiments and case studies. Section \ref{sec:theory} introduces the cake-cutting model and our testable hypotheses. Section \ref{sec:experiment} presents our experimental design. Section \ref{sec:results} discusses our findings. Section \ref{sec:conclusion} concludes.

\section{Related Literature}
\label{sec:literature}

\subsection{Laboratory Experiments}
\label{sub:Previous-Laboratory-Experiments}
All fair division experiments that we know deal with discrete indivisible goods and/or a homogeneous divisible good such as money. This is quite different than our setting, where there is a continuous heterogeneous divisible good. With indivisible goods, the user input usually consists of a ranking of the goods or an assignment of a monetary value to each good. In contrast, cake-cutting has a spatial element --- the participants have to decide where exactly to cut the cake. Since the user interface, user experience and potential manipulations are different, we cannot automatically expect the findings of previous experiments to hold in our setting too. Keeping this in mind, we survey previous lab experiments and compare their findings with ours.

\paragraph{\textbf{Sophisticated versus simple}}
In some experiments, the main research question is \emph{which procedure yields more user satisfaction?} In particular, do users prefer the allocations generated by sophisticated and provably-fair procedures, to the allocations generated by simple and intuitive procedures? 

\citet{Schneider2004Limitations} compare the simple divide-and-choose procedure to the more sophisticated Adjusted-Knaster and Proportional-Knaster procedures, for allocating indivisible goods with monetary compensation.  
They find that, if the participants truthfully adhere to the protocol, then the sophisticated mechanisms perform better than divide-and-choose in terms of efficiency and fairness. However, if the participants are allowed to strategically deviate from the protocol, then their performance declines and becomes comparable to divide-and-choose. 

\citet{DupuisRoy2009Empirical} compare five procedures for indivisible object allocation (Sealed Bid Knaster, Adjusted Winner, Adjusted Knaster, Division by Lottery and Descending Demand) to the allocation with the highest mutual satisfaction scores (which they find using genetic search). They find that the fair division procedures yield allocations that are rather unsatisfactory to humans. They attribute this to two factors which are often ignored by fair division procedures: temporal fluctuation of human preferences, and non-additivity of valuations.

In a different experiment, \citet{DupuisRoy2011Simpler} compare three simple algorithms for allocating indivisible goods (Strict Alternation, Balanced Alternation and Divide-and-Choose) to four provably-fair algorithms (Compensation Procedure, Price Procedure, Adjusted Knaster and Adjusted Winner). They find that, counter-intuitively, the simple algorithms produce significantly fairer allocations.

In contrast, other studies emphasize the advantage of sophisticated fair division procedures. \citet{bassi2006} studied division of homogeneous resources using Crawford's negotiation procedures, and found that his procedures induce even selfish players to act fairly. \citet{gal2017fairest} used the spliddit.org website \citep{goldman2015spliddit}
to study division of rooms and rent, and found that their maximin procedure performs significantly better than a procedure that selects an arbitrary envy-free allocation.

Our findings are in line with the latter studies. Despite the strategic manipulation by humans, the final outcomes of the envy-free procedures (in particular, Selfridge-Conway) are significantly fairer and more satisfactory than those of the non-envy-free procedures. Thus, at least in our setting, the extra-complexity of the procedures pays back in fairness.

\paragraph{\textbf{Strategic manipulation}}
In some experiments, the main goal is to check the strategic behavior of subjects: Do they try to manipulate the protocol? Do they manipulate successfully? And how does the manipulation affect the protocol outcomes?

All previous experiments that we know of found that agents do try to manipulate. However, the effect of this manipulation on the outcome depends on the protocol: in simple auction-based protocols, manipulations lead to highly inefficient outcomes, where no deal was done even though a deal was possible \citep{daniel1998strategic,parco2004enhancing}. Using more structured conflict-resolution procedures (such as Adjusted Winner) did not eliminate manipulation, but it did lead to a much more efficient outcome \citep{Daniel2005Fair,HortalaVallve2010Simple}.

In our experiment, too, we find that subjects try to manipulate the protocol, and the manipulative behavior increases over time. We also find that some procedures are easier to manipulate than others. In particular, Divide-and-Choose and the Knaster-Banach last diminisher procedure are particularly prone to manipulative behavior. 
This might be due to their simplicity --- procedures that are easier to understand are also easier to manipulate.

Strategic behavior was studied extensively in other markets besides fair division, particularly in matching markets \citep{castillo2016truncation} and university-course allocation
\citep{budish2007strategic,budish2012multi}.
A remarkable finding in some of these experiments is that people try to manipulate even when the mechanism is strategyproof, which means that they provably cannot gain by manipulation \citep{parco2004enhancing,artemov2017,hassidim2016,hassidim2017mechanism,rees2017}.

In our experiment this finding is even more pronounced: about 70\% of all subjects tried at least one manipulation that is strictly dominated and obviously results in a smaller payoff for them.

\paragraph{\textbf{Different desiderata}}
In some experiments, the main research question is \emph{what desiderata are more important to users?}
Early experiments checked this question in the simple setting of dividing money (a homogeneous resource). Many experiments check whether, in an inherently unfair game such as the ultimatum game, subjects prefer to accept an unfair offer than to accept nothing \citep{guth1995ultimatum,lopomo2001,guth2003}.

Other experiments check whether, when dividing money among others, people prefer a fair inefficient division to an unfair division that is more efficient
\citep{engelmann2004inequality,fehr2006inequality,herreiner2007distributing}. It was found that such preferences depend on psychological and cultural factors (e.g. economics students choose differently than law students). Later experiments asked similar questions in more complex settings, involving allocation of indivisible objects \citep{herreiner2009,herreiner2010inequality}. These findings are orthogonal to our experiment, in which the fairness desiderata are fixed and the goal is to check which procedure attains them most efficiently.

\subsection{Other experiments}
\paragraph{\textbf{Case studies}}
Besides lab experiments, several fair division procedures were applied to real-life cases. 

\citet{flood1958} studied a case of dividing gift parcels using the Knaster algorithm, and \citet{Pratt1990Fair} applied an auction-based division algorithm to allocate silver heirlooms. They found that, although the algorithm was decentralized and most participants did not fully understand it or the preference information desired, it handled all major considerations well and was regarded as equitable.

Several counter-factual studies checked the feasibility of using the Adjusted Winner (AW) protocol \citep{brams1996} for resolving international disputes, particularly the Camp David Accords \citep{Brams1996Camp}, the Spratly Islands controversy \citep{brams1997fair} and the Israeli-Palestinian conflict \citep{Massoud2000Fair}.

\citet{Tijs2004Cases} describe some case studies of dividing the profits of cooperation between partners, in light of concepts from cooperative game theory, such as the Shapley value.

\citet{kurokawa2015leximin} applied a \emph{randomized leximin mechanism} for allocating public-school classrooms to charter-schools. Unfortunately, the initiator of this algorithm backed away so the mechanism has not been deployed yet, but the partial collaboration emphasized the importance of intuitive and easy-to-understand fairness guarantees.
\citet{oluwasuji2018algorithms} test their heuristic algorithms for \emph{fair load-shedding} on electricity-usage data, which they collected from a USA-based database and adapted to African consumption patterns.

We are not aware of any case studies regarding cake-cutting algorithms. In fact, the only modern application of a cake-cutting procedure that we are aware of is the procedure for allocating areas in the international oceans for mining, which is based on cut-and-choose \citep{young1995equity,Walsh2011Online}:

\begin{quote}

\emph{
``Each application... shall cover a total area... sufficiently large and of sufficient estimated commercial value to allow two mining operations... of equal estimated commercial value. ... The Authority shall designate which part is to be reserved solely for the conduct of activities by the Authority through the Enterprise or in association with developing States... The area designated shall become a reserved area as soon as the plan of work for the non-reserved area is approved and the contract is signed.''
} (United Nations Convention on the Law of the Sea, Annex III, Article 8).

\end{quote}

A possible reason for the rarity of practical use of cake-cutting algorithms may be the lack of data regarding their performance with real people. This is one issue that the present paper aims to improve.

\paragraph{\textbf{Computerized Simulations}}
Computerized simulations of fair division algorithms were used to test properties of such algorithms that are difficult to analyze theoretically. 
\citet{Walsh2011Online} used simulations to compare the welfare properties of online vs. offline  cake-cutting algorithms. 
\citet{Cavallo2012Fairness} 
used simulations to test his mechanism for redistribution of VCG payments.
\citet{dickerson2014computational,aziz2020fair} studied fair allocation of indivisible goods using computerized simulations. They showed that, when the number of goods is sufficiently large (relative to the number of agents), fair allocations are likely to exist. 
Many computerized simulations use the PrefLib library \citep{mattei2013preflib}, which is a collection of real-world preference relations on discrete items.

\subsection{Strategic Fair Division}
There are several theoretical studies regarding the strategic properties of cake-cutting protocols \citep{branzei2013,branzei2016}, and various sophisticated protocols that are truthful under some assumptions on the valuations or agents' behavior: see \cite{nicolo2008}, \cite{mossel2010}, \cite{maya2012}, \cite{chen2013}, \cite{bei2017cake}, and \cite{bei2018truthful} and \cite{ortega2019obvious}.

The repeated-cake-cutting setting has been studied by \cite{delgosha2012information}. They studied ways by which the cutter can exploit her knowledge of the chooser's preferences in order to improve her own welfare. Recently, \cite{tamuz2018non} continued this line of work by suggesting new division protocols that are \emph{non-exploitable}, i.e. a risk-averse cutter cannot improve her welfare using information from previous interactions.

Our work complements these theoretic works in that we study the strategies actually used by human subjects.

\section{Theory}
\label{sec:theory}

We consider a standard setup based on \cite{procaccia2016}. A {\it cake-cutting problem} $([0,1], N, (v_i)_{i \in N})$ is a triplet where:

\begin{itemize}
\item $[0,1]$ is the cake,
\item $N=\{1,\ldots,n\}$ is the set of agents interested in the cake, and
\item $v_i$ is the valuation function of agent $i$, which maps a given subinterval $I \subseteq [0, 1]$ to the value assigned to it by agent $i$, $v_i(I)$.
\end{itemize}

We write $v_i(x,y)$ as a shorthand for $v_i([x,y])$. We assume that $v_i$ satisfies the following standard properties. For every $i \in N$:
\begin{enumerate}
\item For every point $x \in [0,1]$, $v_i(x,x)=0$.
\item For every subinterval $I$, $v_i(I) \geq 0$.
\item For any two disjoint subintervals $I, I'$, $v_i(I) + v_i(I') = v_i(I \cup I')$
\end{enumerate}

We refer to a finite union of disjoint intervals as a {\it piece of cake}. An {\it allocation} $A$ is a partition of $[0,1]$ into $n$ ordered, pairwise-disjoint pieces of cake $A =(A_1, \ldots, A_n)$ such that $A_1 \cup \ldots \cup A_n = [0,1]$. We only consider complete allocations. In a non-strategic framework in which all agents reveal their true valuation function, a {\it procedure} is a function that takes a cake-cutting problem as input and returns an allocation. We normalize the valuation functions so that $v_i(0,1)=1$ for every agent $i$.

\subsection{Division Procedures}
We consider the following procedures to divide a cake among two agents.
\begin{namedthm*}{Asymmetric cut-and-choose}[2ACC]
Agent 1 cuts the cake into two equally-valued pieces, i.e. two pieces $[0,x_1)$ and $[x_1,1]$ such that $v_1(0,x_1) = v_1(x_1,1) = 1/2$. Agent 2 then chooses her preferred piece, and agent 1 receives the remaining piece. Formally, if $v_2(0,x_1) \geq v_2(x_1,1)$, then set $A_2 = [0,x_1), A_1 = [x_1,1]$; otherwise set $A_1 = [0,x_1), A_2 = [x_1,1]$.
\end{namedthm*}

\begin{namedthm*}{Symmetric cut-and-choose}[2SCC] Both agents cut the cake into two equally-valued pieces by choosing $x_i$ such that $v_i(0,x_i)=v_i(x_i,1)=1/2$. Let agent 1 be the one who chooses the lowest cut point $x_1\leq x_2$ without loss of generality. Then, agent 1 receives the piece $A_1=[0,\frac{x_1+x_2}{2})$, and agent 2 receives the piece $A_2=[\frac{x_1+x_2}{2},1]$.
\end{namedthm*}

Both 2ACC and 2SCC have been used and studied since Biblical times (see Genesis 13), yet they are only defined for the division of cake among two agents. We consider three procedures for dividing cake among three or more players. The first of these is the last diminisher procedure suggested by Knaster and Banach.

\begin{namedthm*}{Knaster--Banach last diminisher for $n$ agents}[$n$LD] Given a cake $[y,1]$, agent 1 chooses a cut $x_1$ so that $v_1(y,x_1)=v_1(y,1)/n$. Agent 2 now has the right, but is not obliged, to choose $x_2 < x_1$. Whatever she does, agent 3 has the right, without obligation, to further diminish the already diminished (or not diminished) piece too, and so on up to $n$. The rule obliges the last diminisher (say agent $i$) who chose the cut $x_i$ to take as her allocation $A_i=[y,x_i)$. Agent $i$ is disposed of, and the remaining $n-1$ persons start the same game with the remainder of the cake $[x_i,1]$. When there is only one agent left, she receives the unclaimed piece of cake.
\end{namedthm*}

A similar procedure to $n$LD is the moving-knife mechanism of \cite{dubins1961}, in which agents cut the cake simultaneously rather than sequentially. Here we describe a discrete adaptation of it.
\begin{namedthm*}{Dubins--Spanier for $n$ agents}[$n$DS]
Given a cake $[y,1]$, each agent simultaneously cuts the cake at a point $x_i$ such that $v_i(y, x_i) = v_1(y,1)/n$. The agent $i^*$ who made the leftmost cut exits with the piece $A_{i^*} = [y, x_{i^*}]$. Agent $i^*$ is disposed of, and the remaining $n-1$ persons start the same game with the remainder of the cake $[x_{i^*},1]$. When there is only one agent left, she receives the unclaimed piece of cake.
\end{namedthm*}

An alternative procedure was suggested by \cite{even1984} that improves on $n$LD in that it requires fewer cuts to achieve a proportional allocation.\footnote{The run-time complexity of the Even-Paz procedure is $O(n \log n)$, whereas the one of Knaster--Banach is $O(n^2)$.} The idea of this procedure is to divide the original cake cutting problem into two disjoint ones at each step.
\begin{namedthm*}{Even-Paz for $n$ agents}[$n$EP] For the sake of clarity assume that $n$ is a power of 2. Given a cake $[y,z]$, all agents choose cuts $x_i$ such that $v_i(y,x_i)=v_i(y, z)/2$. We let $x^*$ be the median cut, i.e. the $\left \lfloor{n/2}\right \rfloor$th cut. Then the procedure breaks the cake-cutting problem into two: all agents who choose cuts $x_i \leq x^*$ are to divide the cake $[y,x^*)$, whereas all agents who chose cuts above $x^*$ are to divide the cake $[x^*,z]$.  Each half is divided recursively among the $n/2$ partners assigned to it. When the procedure is called with a singleton set of agents $\{i\}$ and an interval $I$ it assigns $A_i = I$.
\end{namedthm*}

The last three procedures $n$DS, $n$LD and $n$EP can be adapted to divide a cake among any number of agents. Our last procedure is only suitable for dividing cake among 3 agents. It differs from the previous procedures in that it generates allocations that are not contiguous. Furthermore, it requires not one but two cake cuts to be made at the same time.

\begin{namedthm*}{Selfridge-Conway}[3SC]
Agent 1 cuts the cake into three pieces of equal value to her: $I_1, I_2, I_3$; so that $v_1(I_i)=1/3$. Agent 2 divides the piece of highest value to her, say $I_1$ into $I_1'$ and $T=I_1 \setminus I'_1$, so that the value of $I_1'$ is the same of the second most valuable piece, say $I_2$: $v_2(I'_1)=v_2(I_2)$. We separate the original cake into the modified cake $C'=C \setminus T$ and the trimmings $T$. First we allocate $C'$. Let agent 3 choose and take her favorite piece among $I'_1,I_2,I_3$. If she chooses $I'_1$, let agent 2 choose any remaining piece; but if agent 3 chooses $I_2$ or $I_3$, then give $I'_1$ to agent 2 without letting her choose. Agent 1 receives the leftover piece.

Now we assign $T$. Let $i \in \{2,3\}$ be the player who obtained $I'_1$, and $j$ the other one. Agent $j$ splits $T$ into three parts of equal value to her. Now agent $i$, 1, and $j$ choose a piece of $T$ in that specified order.
\end{namedthm*}

\subsection{Fairness Properties}
We consider the following fairness properties of allocations.

\begin{definition}An allocation $A$ is {\it proportional} if each agent gets at least $1/n$ of the cake according to her own evaluation, i.e. if $\forall i\in N: v_i(A_i) \geq 1/n$.
\end{definition}

\begin{definition}An allocation $A$ is {\it envy-free} if no agent prefers another agent's share, i.e.
$\forall i,j\in N: v_i(A_i) \geq v_i(A_j)$.\footnote{These notions should not be confused with {\it procedural envy-freeness} or anonymity, which requires that the procedure treats agents symmetrically \citep{nicolo2008, bhardwaj2020fairness}. }
\end{definition}

In our setup, envy-freeness implies proportionality, while the converse is true only for the case of two agents. A procedure is envy-free or proportional if, for every cake-cutting problem, it produces an allocation that is envy-free or proportional, respectively.
The following lemma summarizes the well-known fairness properties of these procedures; see \citet{robertson1998} for proofs.
\begin{lemma} 2ACC, 2SCC, $n$DS$, n$LD, $n$EP and 3SC are all proportional. 2ACC, 2SCC, and 3SC are envy-free. $n$DS, $n$LD, and $n$EP are not envy-free.
\end{lemma}

The previous lemma gives us our first hypothesis, which refers to dividing a cake among three agents.
\begin{hypothesis}
\label{hypo:fair}
Allocations received under 3SC are {\it fairer} than those received under 3DS and 3LD, i.e. generate fewer cases of envy.
\end{hypothesis}

\subsection{Incentive Properties}
Another important goal of cake-cutting procedures is to give incentives to agents to reveal their true (privately known) valuation function to a mediator who, after receiving the report from all agents, conducts a division procedure. The valuation function is partially revealed via a series of cake cuts or choices between pieces of cake. Although the mediator does not know the valuations, it is assumed that every agent knows the other agents' valuations.

In a strategic framework, given a cake $[0,1]$ and a set of agents $N$, a procedure $p$ is a function from the revealed valuation function of each agent to an allocation $A$. We write $p_i(v_i, v_{-i})=A_i $ to denote the cake allocated to agent $i$ by procedure $p$, where $v_{-i}$ denotes the reported valuation functions of all other agents except $i$. We use the following standard property to study which procedures are robust to strategic behavior.
\begin{definition}The procedure $p$ is {\it strategy-proof} if for every agent $i$, every collection of valuation functions $(v_i,v_{-i})$, and every other valuation function $v'_i$,
\begin{equation}
v_i(p_i(v_i, v_{-i})) \geq v_i(p_i(v'_i,v_{-i}))
\end{equation}
\end{definition}

Note that the definition is rather demanding: a procedure is strategy-proof only if behaving non-strategically is a dominant strategy for every player.\footnote{This is the standard notion of strategy-proofness in mechanism design. For weaker notions see \cite{brams2006,brams2008} and \cite{ortega2019obvious}.} Therefore, it is not surprising that:
\begin{lemma}
\label{thm:strategyproof}
2ACC, 2SCC, $n$DS$, n$LD, $n$EP and 3SC are all not strategy-proof.
\end{lemma}

Lemma \ref{thm:strategyproof} is also well-known; \cite{brams2006} in particular discuss many examples of how all these procedures can be manipulated.

A related question is how much agents can gain by acting strategically compared to their guaranteed payoff obtained with non-strategic behavior in any of the procedures we have described.\footnote{We use the term ``non-strategic'' for what often in the literature is called truthful behavior (e.g. \citealp{chen2013}). All the procedures that we consider are often studied as a series of cut-and-evaluate queries that a mechanism designer asks to agents; the so-called Robertson--Webb framework. Every manipulation can be associated to an insincere answer to a query in the Robbertson--Webb framework. Non-strategic behavior may be referred to as truthful or straightforward, as a referee pointed out.} We answer this question by considering the notion of $\epsilon$-strategy-proofness, which has recently been suggested in the literature \citep{Menon2017}. In layman terms, a cake-cutting procedure is $\epsilon$-strategy-proof if there is no cake-cutting problem for which a misrepresentation of preferences guarantees more than $\epsilon$ utility compared to non-strategic behavior.\footnote{Formally, for any $\epsilon \in [0,1]$, the procedure $p$ is \emph{$\epsilon$-strategy-proof} if for every agent $i$, every collection of valuations functions $(v_i, v_{-i})$, and every other valuation function $v_i'$, $v_i(p_i(v_i, v_{-i}))\geq v_i(p_i(v_i', v_{-i}))-\epsilon$.
} A proportional procedure should ideally have  $\epsilon=0$, and in the worst case $\epsilon=1-\frac{1}{n}$: this means that non-strategic behavior guarantees an agent $\frac{1}{n}$, whereas lying yields the maximum utility possible (1). Unfortunately, we show that all the procedures we consider can offer the largest incentives for preference manipulation.

\begin{proposition}\label{lem:max-strategising-payoff}
The procedures 2ACC, 2SCC, $n$DS, $n$LD, $n$EP, 3SC are $\left(1-\frac{1}{n}\right)$-strategy-proof and this is tight.
\end{proposition}

We postpone the constructive proof of Proposition \ref{lem:max-strategising-payoff} to the Appendix.

These two results suggest that if agents know their opponents' preferences in real-life cake-cutting, they should behave strategically if the cake-cutting problem admits a successful manipulation. This is our second hypothesis.
\begin{hypothesis}
\label{hypo:strategy1}
Agents who know their opponents' preferences behave strategically in 2ACC, 2SCC, $n$DS$, n$LD, $n$EP and 3SC.
\end{hypothesis}
The assumption that agents know their partners' valuations is a strong one, yet necessary for agents to manipulate the procedure to their advantage with certainty of success. 
Without such knowledge, an agent might perform a manipulation that will decrease her utility. 
Therefore, in the fair division literature, it is often claimed that strategic manipulation is not an issue when people do not know their partners' preferences (see e.g. \citet{gal2017fairest}).

However, in real life, agents may have a partial knowledge about their partners' preferences, particularly if they have interacted with those partners previously. In those cases, an agent is able to learn the other agents' valuations through experimentation, i.e. choosing different strategies each interaction in order to eventually improve their own allocation. This simple observation provides us with our final hypothesis.
\begin{hypothesis}
\label{hypo:strategy2}
Agents who do not know their opponents' preferences but who repeatedly interact with them, successfully learn their opponents' preferences and do behave strategically in 2ACC, 2SCC, 3SC, $n$DS$, n$LD, $n$EP.
\end{hypothesis}

\subsection{Fairness and Incentives}
It is important to note a dependency between our three hypotheses.
Hypothesis \ref{hypo:fair} states that 3SC is fairer than 3LD and 3DS since it generates envy-free allocations. However, this envy-freeness is guaranteed only when all agents report their preferences non-strategically. 
In contrast, Hypotheses \ref{hypo:strategy1} and \ref{hypo:strategy2} state that people behave strategically. 
If all agents behave strategically, then in general, all three procedures discussed --- 3SC, 3LD and 3DS --- generate envy \citep{branzei2016}.
However, hypothesis \ref{hypo:fair} still holds if the procedures are used by a population in which a fraction $\alpha$ of agents behave non-strategically. Then 3SC guarantees envy-freeness in at least $\alpha$ cases, and thus it is reasonable to expect that it would still be perceived as fairer than 3LD and 3DS. As a consequence, Hypothesis \ref{hypo:fair} extends to cases in which a constant fraction of the agents behave non-strategically.\footnote{The fraction $\alpha$ is in fact not constant but specific to each procedure. However, from our lab experiments we found that the fraction of agents who behave non-strategically in 3SC is larger than in 3DS and 3LD, and thus it is safe to expect that 3SC is indeed perceived fairer than 3DS and 3LD.}

A related interesting question is whether envy can be generated in 3SC when only one agent misreports her preferences, while the other agents are non-strategic. This question is particularly relevant to our experimental setting in EXP1, since the computerized agents are non-strategic so only the single human subject might act strategically.\footnote{\cite{branzei2016} prove that envy can be generated in Nash equilibrium of 3SC, but their proof crucially relies on the assumption that all three agents behave strategically.
}
We answer this question in the affirmative by showing that:
\begin{proposition}\label{lem:envy-when-strategizing}
Envy can be generated in 3SC with just one agent misrepresenting her preferences. This agent achieves a higher payoff at the cost of being envious.
\end{proposition}
The proof is postponed to the Appendix. 

\section{Experiment}
\label{sec:experiment}
In this section we present the set-up of our two experiments in detail. We begin with their design and proceed to the implementation.
\subsection{Design}
We conducted two experiments in which subjects divide cakes using the aforementioned procedures. In the first one (EXP1), subjects divide a cake against automated non-strategic agents. Subjects are told that they will divide cake against agents, but they do not know that the agents are automata that behave non-strategically. In the second experiment (EXP2), 
subjects divide a cake against other subjects (real people who behave strategically). In EXP1, subjects divide cakes against 1, 2 and 3 opponents, whereas in EXP2 subjects only divide cakes against 1 or 2 opponents (because EXP2 is significantly more time consuming for subjects, since they have to wait for their peers' decisions). The procedures used in EXP1 and EXP2 are described in {\bf Table \ref{tab:procedures}} below. In EXP1, procedures are presented in a fixed order so that subjects face the easiest procedures first, whereas in EXP2 the procedures are presented randomly, to exclude the possibility of order effects (see {\bf Table \ref{tab:order}}).

\begin{table}[!htbp]
\caption[caption,justification=centering]{Summary of cake-cutting procedures in EXP1 and EXP2.}
\label{tab:procedures}
\centering
\begin{tabular}{l|ccccc}
	\toprule
	Procedure                      & Agents & EXP1 & EXP2 & Envy-free & Prop \\
	\midrule
	Cut and choose (2ACC)                & 2      & \ding{51}    & \ding{51}    & \ding{51}         & \ding{51}            \\
	Cut middle          (2SCC)           & 2      & \ding{51}    & \ding{51}    & \ding{51}         & \ding{51}            \\
	Selfridge--Conway  (3SC)  & 3      & \ding{51}    & \ding{51}    & \ding{51}         & \ding{51}            \\
	Dubins--Spanier ($n$DS) & $n$      & \ding{51}($n=3,4$)  & \ding{51}($n=2,3$)  & $\times$        & \ding{51}            \\
	Last Diminisher ($n$LD) &$n$      & \ding{51}($n=3,4$)    & \ding{51}($n=3$)    &     $\times$   & \ding{51}            \\
	Even--Paz   ($n$EP)                    & $n$      & \ding{51} (4)&$\times$  & $\times$        & \ding{51}           \\
	\bottomrule 
\end{tabular}
\end{table}

\begin{table}[!htbp]
\caption[caption,justification=centering]{Order of cake-cutting procedures in EXP1 and EXP2.}
\label{tab:order}
\centering
\begin{tabular}{l|c}
	\toprule
&Order of procedures\\
	\midrule
	EXP1&2ACC, 2SCC, 3DS, 4DS, 3LD, 4LD, 4EP, 3SC\\
	EXP2&Random \\
	\bottomrule
\end{tabular}
\end{table}

We change the names of the procedures to make it easier for the subjects to understand them. In EXP1, we use the following names: {\it I Cut You Choose} (for 2ACC), {\it Cut Middle} (for 2SCC), {\it Leftmost Leaves} (for $n$DS), {\it Last Challenger} (for $n$LD), {\it Super Fast} (for 4EP) and {\it Super Fair} (for 3SC). In EXP2, we change the name of {\it Super Fair} to {\it Double Knife} to avoid experimenter demand effects.
	
Each cake is divided 7 times. We call each of these divisions a {\it round}. During the first five rounds, the subjects don't know their opponents' valuations. In the remaining two rounds, the subjects observe their opponents' valuations. We give subjects 5 rounds to experiment and learn their opponents' valuations. The valuations of the subjects (and the automata) are constant during the 7 rounds in each procedure, but they change across procedures. In EXP1, the subject makes the first cut in all procedures that are sequential. In EXP2, the roles of agents are assigned at random in each procedure, but remain constant during the seven rounds.

In all procedures, the cake is a line and the subjects' and automated agents' valuations are normalized so that $v_i(0,1)=120$. In other words, the subject and the automata can obtain a maximum of 120 points if they obtain all of their desired parts of the cake. We chose 120 because it is easily divisible by 2, 3, and 4. Subjects are shown their valuations on the computer screen. Their valuations are given by a set of subintervals which are deemed desirable, while the rest of the $[0,1]$ interval is not (valuations are the same in EXP1 and EXP2, and are presented in the Appendix). All desired intervals of the same length yield the same payoff; such valuations are known as {\it piecewise uniform}. The cake can only be cut in a position $x$ so that $v_i(0,x)$ equals an integer number between 0 and 120. A  representative screen that subjects observe during EXP1 and EXP2 is shown in {\bf Figure \ref{fig:interface}}.\footnote{The graphical interfaces can be downloaded from our website. EXP2 is implemented using the o-Tree software \citep{chen2016otree}. It can be played online at \url{https://cakecut.herokuapp.com/demo/}.
}

\begin{figure}[!htbp]
	\centering
		\includegraphics[width=.7\textwidth]{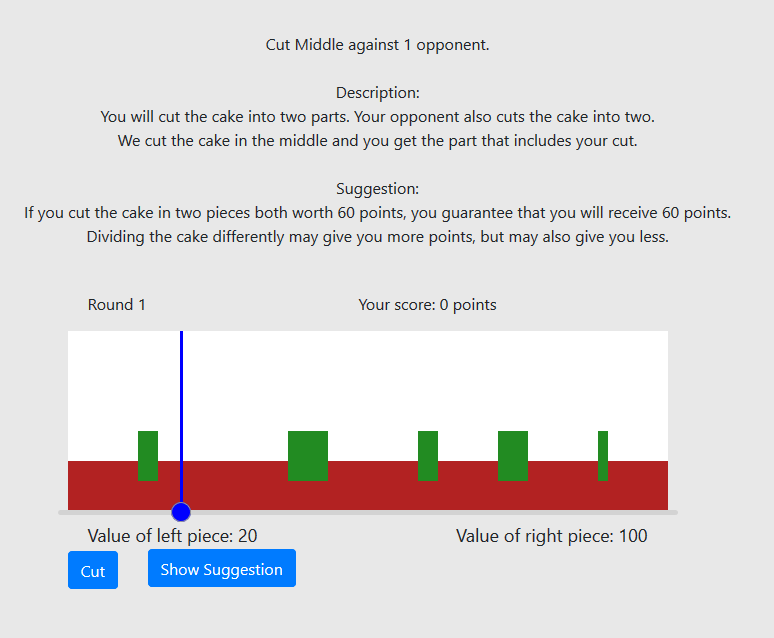}
	\caption{An example of our experimental interface. The cake is depicted as a brown line, while the desirable parts of the subject are emphasized with a green color.}
	\label{fig:interface}
\end{figure}

After completing each round, subjects are told what share of cake they got, and the valuation of the shares that the other players received, calculated by their own valuation function. We chose the valuations so that strategic behavior yields substantial benefits over non-strategic behavior (these are provided in the Appendix). The subjects are given the suggestion to cut the cake non-strategically, but are also explicitly told that they can choose another strategy that may give them more or less points than the non-strategic one. Ties are broken according to agents’ index in the procedure: in EXP1 ties are always broken in favor of the human subject, whereas in EXP2 they are broken according to the index of the participants, which is assigned randomly. 

The payment in EXP1 was in GBP, whereas in EXP2 was in EUR. In both EXP1 and EXP2, the highest payment achievable in EXP1 and EXP2, through strategic behavior, is 29 currency units, whereas the lowest is 5 currency units (i.e either GBP or EUR), which subjects receive for showing up. In EXP1, in addition to the 5 currency unit payment for showing up,  2 rounds are randomly selected from all procedures, and subjects are paid the number of points they obtained in both procedures divided by 10. For example, if in the two randomly selected rounds the subject obtains 120 and 80 points, then she receives 12 GBP + 8 GBP +  GBP 5 = 25 GBP. In EXP2, in addition to the 5 currency unit payment for showing up, the subjects are paid the average payoff across rounds divided by 5. 

In EXP2, before the experiment begins, we give the subjects one practice round that is not relevant to their payoffs so that they familiarize themselves with the graphical interface and the procedures.

\subsection{Implementation}
EXP1 was conducted at the EssexLab facilities at the University of Essex during July 2018. EXP2 was conducted at the mLab facilities at the University of Mannheim and at the AWI-Lab at the University of Heidelberg during September and November 2019. In both experiments, most of the experiment participants were undergraduate students.

Upon their arrival to the lab, subjects were randomly assigned to a computer. They signed a consent form and were given the experiment instructions in a short presentation by the principal investigator (these are provided in the Appendix). They were allowed to ask questions during and after the instructions were given. After all questions had been answered, the subjects were allowed to start the experiment. Subjects were not allowed to communicate with other subjects during both experiments (except possibly through their actions on the platform). In EXP1, the role of the participant in the procedure is that of the first cutter.  In EXP2, subjects are assigned to groups (of two or three) with others randomly in each procedure, and their role in each procedure is assigned randomly, and maintained during the entire 7 rounds. After the experiment ended, subjects were paid in private and dismissed.
\begin{figure}[t]
	\centering
	\begin{subfigure}{.45\linewidth}
		\includegraphics[width=.85\textwidth]{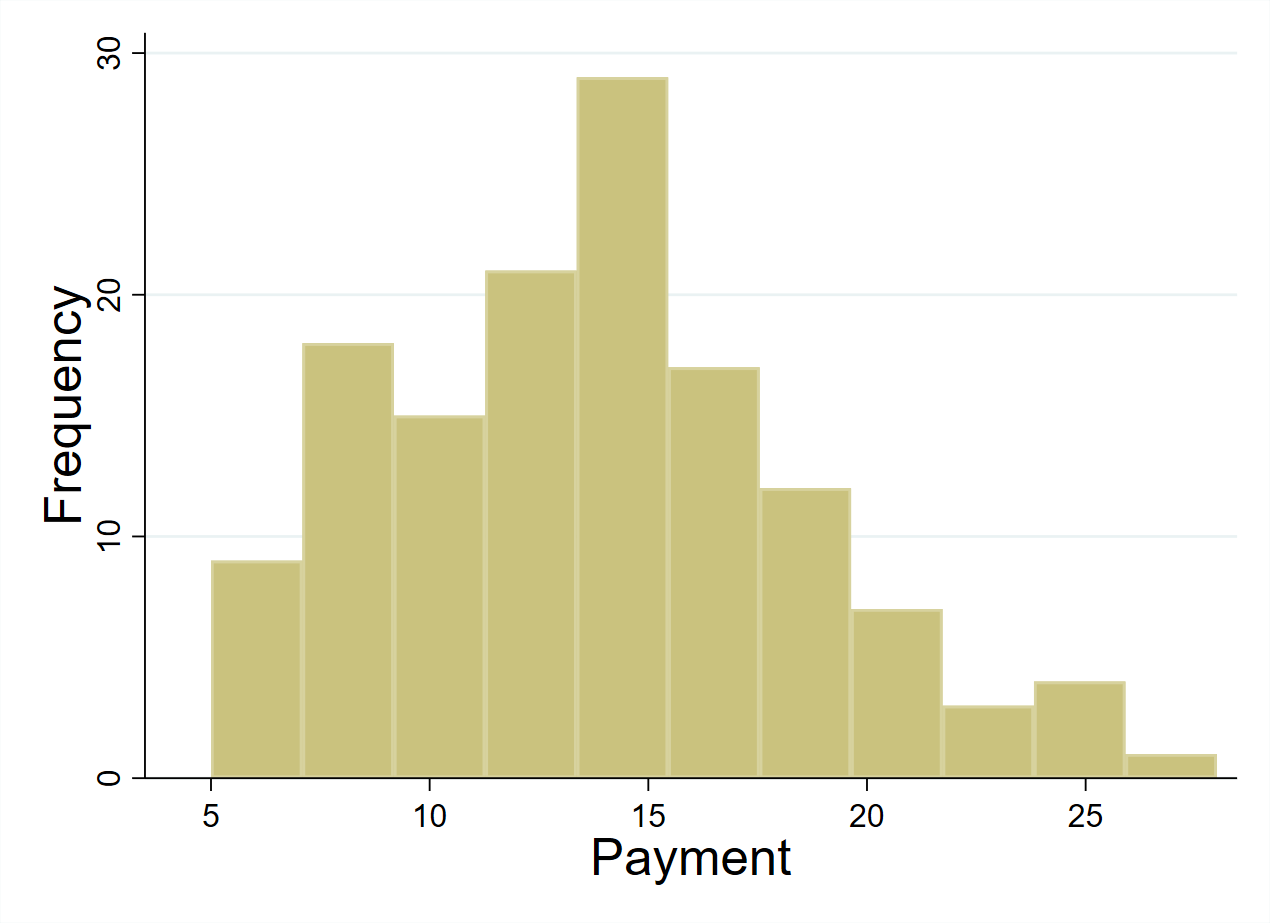}
		\caption{EXP1.}
		\label{fig:histogramessex}
	\end{subfigure}
	\begin{subfigure}{.45\linewidth}
		\includegraphics[width=.85\textwidth]{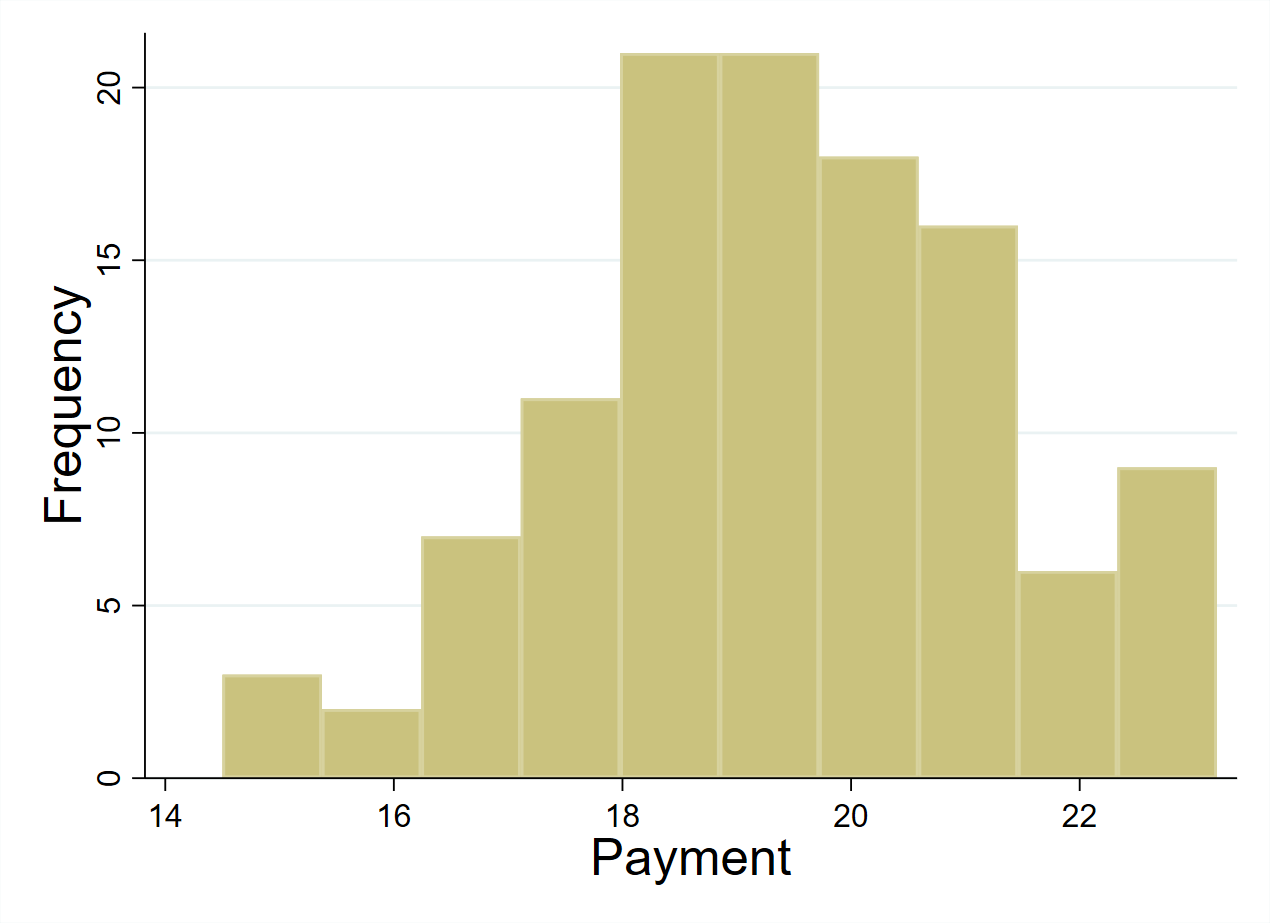}
		\caption{EXP2.}
		\label{fig:histogrammannheim}
	\end{subfigure}
	\caption{Distribution of Payments.}
	\label{fig:histogram}
\end{figure}

We had 136 participants in EXP1 and 114 participants in EXP2 (in EXP1 we lost 3 observations due to technical problems so our final participant pool size is 133; for a detailed split to sessions see Table \ref{tab:sessions}). The distribution of payments is shown in {\bf Figure \ref{fig:histogram}}. 

\begin{table}[!htbp]
\caption[caption,justification=centering]{Summary of lab sessions.}
\label{tab:sessions}
\centering
\begin{tabular}{l|cc}
	\toprule
	             & \multicolumn{2}{c}{Number of subjects} \\
	                      & EXP1 & EXP2 \\
	\midrule
	Session 1                & 9     & 18\\
	Session 2                & 12     & 18\\
	Session 3                & 24     & 12\\
	Session 4                & 15     & 12\\
	Session 5                & 27     & 12\\
	Session 6                & 26     & 18\\
	Session 7                & 23     & 12\\
		Session 8                & -    & 12\\
			Session 9                & -    & 12\\
	Total&136 (133 full observations) &114 \\
	\bottomrule
\end{tabular}
\end{table}

\section{Results}
\label{sec:results}
In this section we present and discuss the findings of our experiments. We provide results with respect to the extent of manipulation, envy and fairness, opportunities for learning, and quality of learning,  and present them in separate subsections. We note that both experiments support the same claims which increases our confidence in the results.

\subsection{Manipulation}
\label{sub:results-manipulation}

\begin{figure}[!htbp]
	\centering
	\begin{subfigure}{.495\linewidth}
		\centering
		\includegraphics[width=.95\textwidth]{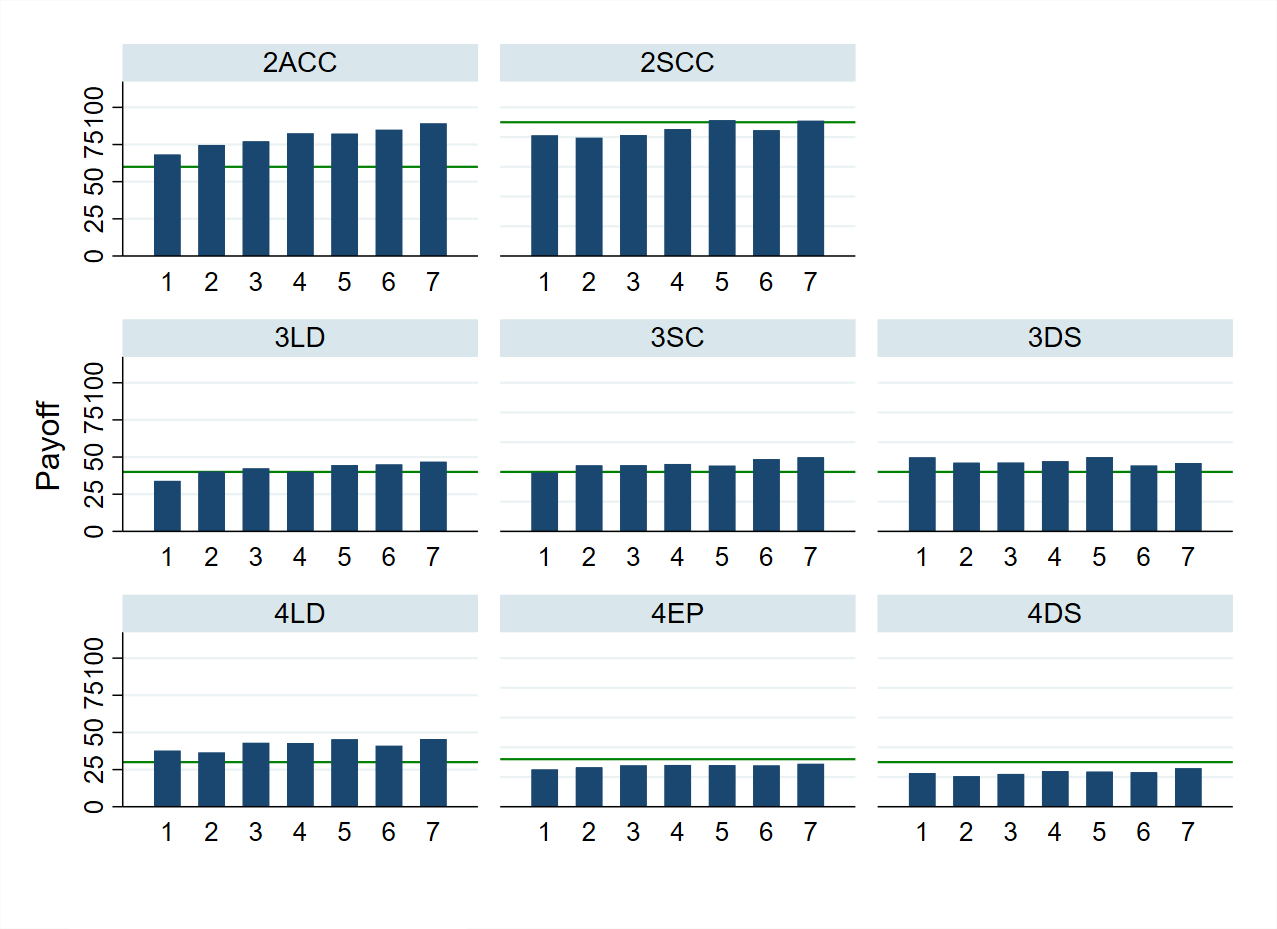}
		\caption{EXP1}
		\label{fig:comparison1}
	\end{subfigure}
	\begin{subfigure}{.495\linewidth}
		\centering
		\includegraphics[width=.95\textwidth]{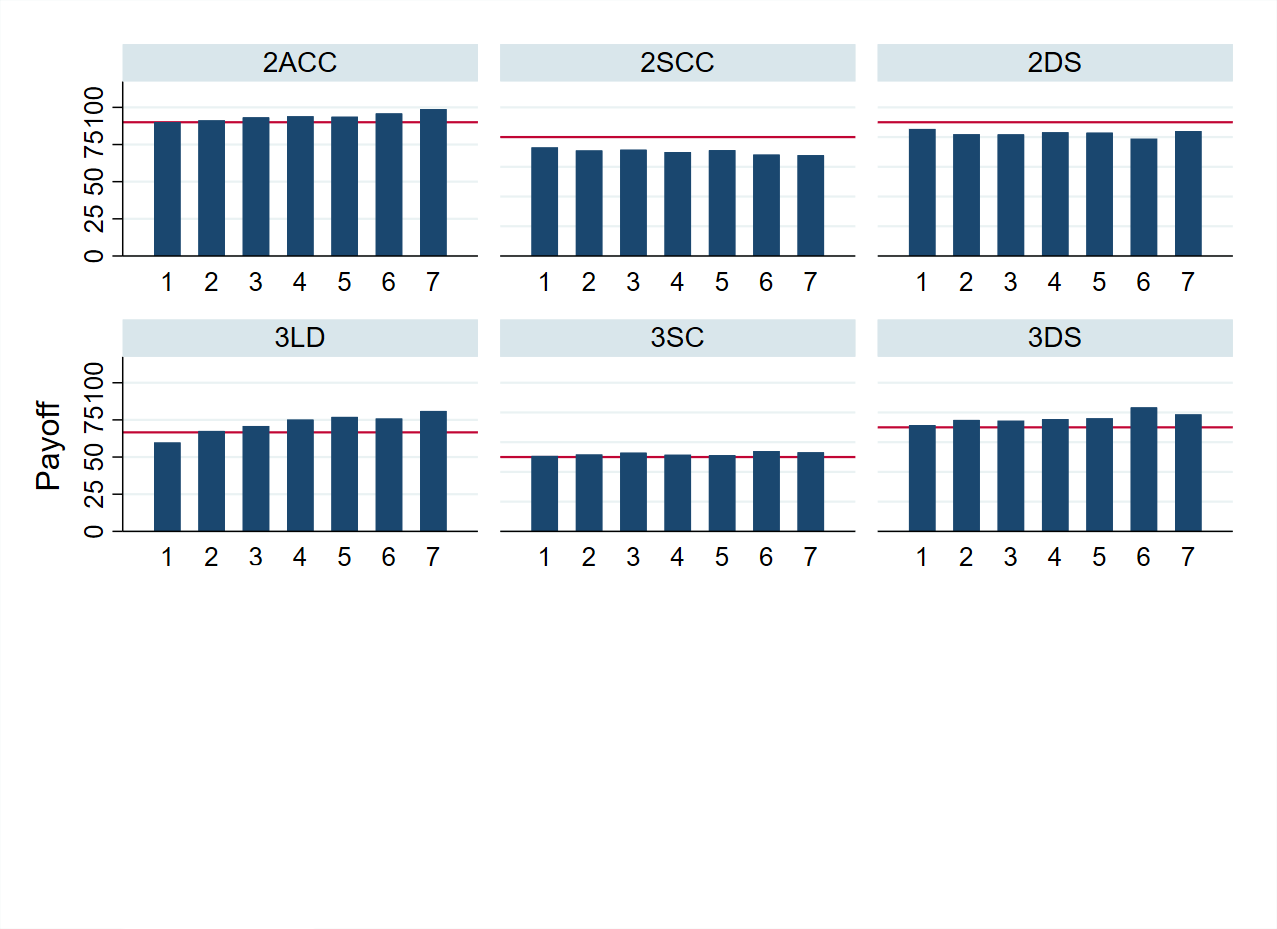}
		\caption{EXP2}
		\label{fig:comparison2}
	\end{subfigure}
	\caption{Average points obtained by the subjects in each round, by procedure. The green line in Fig. 3(a) corresponds to the payoff of the human subject when acting non-strategically. The red straight line in Fig. 3(b) corresponds to average payoff across all participants when they all behave non-strategically.
}
	\label{fig:comparison}
\end{figure}

The findings with respect to manipulation and learning can be previewed in {\bf Figure \ref{fig:comparison}}, which presents the average number of points obtained in each round, by procedure.

{\bf Figure \ref{fig:comparison}} makes evident that subjects manipulate the procedures, even before they have any information about their opponents' preferences. In the majority of procedures, the average payoff in the first round is lower than the one that could be obtained with non-strategic behavior, showing that agents manipulate the procedure even without having information about their opponents' preferences. In some other few cases (2SCC  in EXP1, 2ACC and 3DS in EXP2), we observe the opposite: that those early manipulations lead to higher average payoffs. Overall, payoffs increase in later rounds (as we describe formally in the next subsections), but learning does not always lead to higher payoffs (e.g. 2SCC in EXP2), and even when we observe improvements from learning, agents sometimes still would achieve higher payoffs if they had behaved non-strategically (e.g. 3LD in EXP1). The percentage of manipulations that yield a lower payoff than that guaranteed by proportionality (i.e. $\frac{120}{n}$) is 19\%, 32\% and 56\% for procedures with 2, 3 and 4 players in EXP1, and 19\% and 12\% for procedures with 2 and 3 players in EXP2.

{\bf Table \ref{tab:refrequest}} presents the average payment of strategic and non-strategic agents in each procedure, showing that the payoff for strategic agents is significantly different from that of non-strategic agents for most procedures in EXP1 and EXP2 (except in 3LD in EXP1 and in 2DS, 3DS and 3SC in EXP2).

{\bf Figure \ref{fig:truthful}} presents the percentage of non-strategic cuts/choices observed.  Because in some procedures some agents may be required to make more than one cut, we examine only the first cut.\footnote{We give a $\pm$ 5 pixel tolerance interval when defining non-strategic cake cuts to allow for mistakes. We conducted several robustness checks changing the tolerance to 0 and $\pm$ 10 pixels. The results were almost identical.} In the first round, non-strategic behavior exceeds strategic behavior in only two procedures in EXP1 (2SCC and 3SC) and only two procedures in EXP2 (2ACC and 3SC).\footnote{The large amount of non-strategic behavior in 2ACC in EXP2 is due to the half of the participants who only choose among cut pieces and thus could never benefit from manipulation.} We observe higher rates of non-strategic behavior in envy-free procedures. We compare the difference between the percentage of non-strategic behavior observed in three-agent envy-free procedures (3SC) versus three-agent proportional procedures (3DS, 3LD). The difference is of 17 and 29 percentage points for EXP1 and EXP2, respectively, and is statistically significant in both cases. To see this, we perform a cluster-adjusted t-test (clustering at the subject level in EXP1 and at the subject and session level in EXP2)\footnote{We do not cluster at the session level in EXP1 because subjects play against automata and never interact with other participants in their session.}. In all cases, the corresponding p-value is smaller than 0.001.
\begin{figure}[!htbp]
	\centering
	\begin{subfigure}{.495\linewidth}
		\centering
		\includegraphics[width=.95\textwidth]{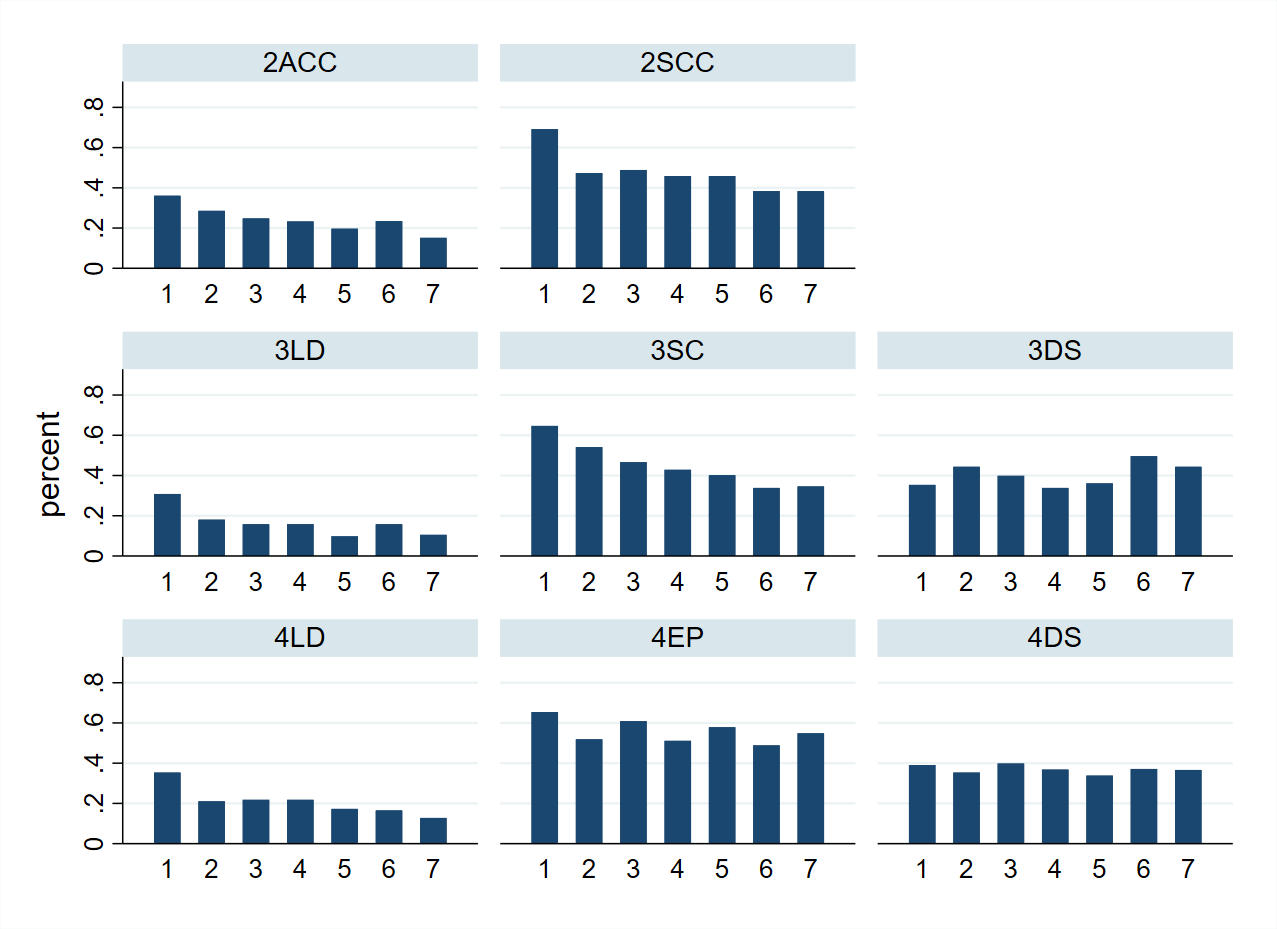}
		\caption{EXP1.}
		\label{fig:comparison14}
	\end{subfigure}
	\begin{subfigure}{.495\linewidth}
		\centering
		\includegraphics[width=.95\textwidth]{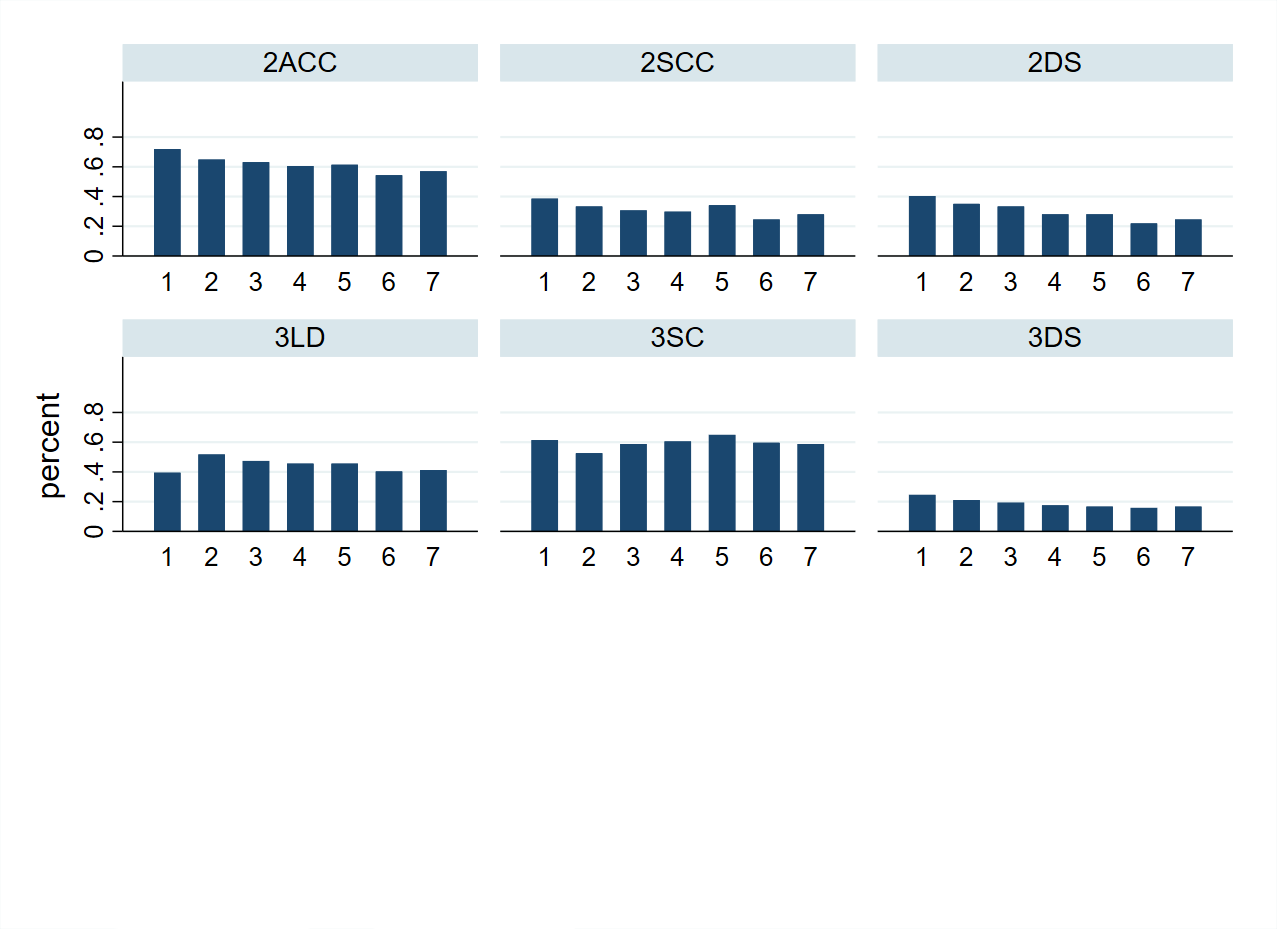}
		\caption{EXP2.}
		\label{fig:comparison24}
	\end{subfigure}
	\caption{Percentage of non-strategic cake cuts/choices, by procedure over rounds.}
	\label{fig:truthful}
\end{figure}

\begin{table}[h!]
	\begin{center}
		\caption[caption,justification=centering]{Average points obtained by sincere and strategic agents.}
		\label{tab:refrequest}
		\resizebox{\textwidth}{!}{
			\begin{tabular}{llllllllll}
				\toprule
				& 2ACC 		& 2SCC & 2DS & 3DS & 3LD & 3SC & 4DS & 4EP & 4LD  \\
				\midrule
				\multicolumn{10}{l}{\it EXP1}                                                        \\
				Non-strategic & 60          & 90         &  $\cdot$    & 40   & 40  & 40  & 30  & 32  & 30  \\
				Strategic  & 86          & 80         &  $\cdot$    & 52   & 42  & 49  & 19  & 22  & 45  \\
				Difference  & 26         & -9         &   $\cdot$  & 12    & 2   & 9   & -11   & -10    & 15   \\
				\footnotesize p-value* & \footnotesize 0.00  & \footnotesize 0.00  &   $\cdot$  & \footnotesize 0.00   & \footnotesize 0.46 & \footnotesize  0.00 & \footnotesize 0.00   & \footnotesize  0.00 & \footnotesize   0.00     \\
				\midrule
				\multicolumn{10}{l}{\it EXP2}                                                        \\
				Non-strategic & 98          & 73         & 84   & 75   & 79  & 53  &  $\cdot$   &  $\cdot$   &  $\cdot$   \\
				Strategic  & 88          & 69         & 82   & 77   & 68  & 51  &  $\cdot$   & $\cdot$   &  $\cdot$   \\
				Difference  &   -10         & -4        &   -2  & 2   & -11  &  -2 & $\cdot$   & $\cdot$    & $\cdot$    \\
				\footnotesize p-value* & \footnotesize 0.01 & \footnotesize 0.04  &   \footnotesize 0.52  & \footnotesize  0.72  & \footnotesize 0.01 & \footnotesize 0.28  & $\cdot$    & $\cdot$   & $\cdot$        \\
				\footnotesize p-value** & \footnotesize 0.02 & \footnotesize 0.03  &   \footnotesize 0.42  & \footnotesize  0.64  & \footnotesize 0.00 & \footnotesize 0.10  & $\cdot$    & $\cdot$   & $\cdot$        \\
				\bottomrule
			\end{tabular}
		}
		\begin{tablenotes}
			\item \footnotesize We report the p-value for a cluster-adjusted t-test testing the null hypothesis that the difference is zero. One asterisk indicates standard errors clustered at the subject level, two indicate clustering at the session level. Rounding errors sometimes cause the difference to not match the original values exactly (up to 1 digit).
		\end{tablenotes}
	\end{center}
\end{table}

We observe that strategic behavior increases in later rounds after subjects learn their opponents' preferences. Indeed, we conduct a logit regression of the probability of playing non-strategically on the round number. The obtained coefficients are -0.1 and -0.07 for the first and second experiments, respectively and are statistically significant (p-value in both cases $<$ 0.001, see {\bf Table \ref{tab:regression1}}). Thus, we conclude that:
\begin{table}[!h]
	\begin{center}
		\caption[caption,justification=centering]{Probability of non-strategic play explained by round number.}
		\label{tab:regression1}
\begin{tabular}{lcccc}
	\toprule
	& \multicolumn{4}{c}{Dependent Variable:}    \\
	& \multicolumn{4}{c}{Prob. of   non-strategic play}    \\
\cline{2-5}
	& \multicolumn{2}{c}{EXP1} & \multicolumn{2}{c}{EXP2}      \\
\cline{2-5}
	Round number   & -0.10      & -0.10   	& -0.07   & -0.07       \\
	Standard error & 0.013    & 0.014   	& 0.015   & 0.146      \\
	Subject FE     & No          & Yes      & No     & Yes    \\
	Session FE     & -            & -         & No     & Yes     \\
	Observations   & 7,441       & 7,441    & 4,746  & 4,746\\
	P-value        & 0.000       & 0.000    & 0.000   & 0.000       \\
	Clusters       & 133         & 133      & 113    & 8      \\
	\bottomrule
\end{tabular}

{\footnotesize	Logistic regression. Standard errors clustered at the subject level in the first three columns and at the session level in the fourth column. Removing subject fixed effects (FE) changes the results minimally.}
\end{center}

\end{table}

\begin{result}
	Subjects manipulate (often unsuccessfully) all the division procedures. Envy-free procedures are significantly less manipulated than proportional ones. Strategic behavior increases with learning.
\end{result}

Our findings are in line with those of \cite{HortalaVallve2010Simple} who, in a different fair division procedure, in which two subjects vote for a series of issues, also document strategic behavior increasing with learning over time.

\subsection{Envy and Fairness}
\label{sub:results-envy}

Envy emerges in all of the division procedures, even envy-free ones, although at quite different rates. The percentage of cases in which envy emerges in each procedure is summarized in {\bf Table \ref{tab:envy}}.
\begin{table}[!h]
\begin{center}
\caption[caption,justification=centering]{Percentage of cases where envy is generated, by round.}
\label{tab:envy}
\resizebox{\textwidth}{!}{
\begin{tabular}{l|lllll|l|ll|l|l}
	\toprule
	$\qquad$  Round	& 1  & 2  & 3  & 4  & 5  & No knowl. & 6  & 7  & Knowl. & Total \\
	Proc &&&&&&average&&&average&average\\
	\midrule
	\multicolumn{10}{l}{\it EXP1}\\
	2ACC & 10 & 8  & 8  & 5  & 8  & 8   & 5  & 3  & 4   & 7     \\
	2SCC & 16 & 25 & 23 & 20 & 13 & 19  & 17 & 14 & 16  & 18    \\
	3DS  & 57 & 68 & 62 & 62 & 59 & 62  & 74 & 68 & 71  & 64    \\
	3LD  & 56 & 53 & 51 & 53 & 43 & 51  & 44 & 40 & 42  & 49    \\
	3SC  & 31 & 31 & 29 & 29 & 32 & 30  & 23 & 25 & 24  & 29    \\
	4DS  & 64 & 86 & 84 & 78 & 77 & 78  & 78 & 73 & 76  & 77    \\
	4LD  & 62 & 66 & 53 & 54 & 50 & 57  & 53 & 44 & 49  & 55    \\
	4EP  & 97 & 92 & 94 & 93 & 92 & 94  & 91 & 89 & 90  & 93   \\
	\midrule
	\multicolumn{10}{l}{\it EXP2}\\
		2ACC & 4&	4&	4&	4&	5&	4&	5&	2&	3.5 &	4\\
	2SCC & 17&	25&	24&	29&	20&	23&	26&	28&	27&	24\\
	2DS & 8&	12&	13&	11&	12&	11&	19&	11&	15&	12\\
	3DS  & 25	&24&	25&	26&	25&	25&	17&	24&	21&	24\\
	3LD  & 34&	25&	20&	15&	12&	21&	16&	9&	13&	19\\
	3SC & 23 &25& 22&16 &14&	18&	11&	12&	12&	16		\\
	\bottomrule
\end{tabular}
}
\end{center}
\end{table}

Envy may emerge in envy-free procedures due to two reasons. One is that subjects strategically manipulate their cake cuts. Another is that subjects do not understand the procedure. Even in 2ACC, the simplest of the procedures, envy was generated in 3\% and 2\% of the cases in the last round  of EXP1 and EXP2, respectively, when subjects knew their opponents' preferences and were already familiar with the division procedure. In these cases, envy was generated by thoughtless cake cuts. The data for 2ACC suggest that this dull behavior occurs rarely. Most of the envy is instead caused by strategic experimentation of the subjects, and reduces once subjects know their opponents' preferences, in the last two rounds. It is somewhat surprising that envy is generated in a considerable fraction of cases in 2SCC, which we observe is due to the fact that subjects follow the simple heuristic of copying a manipulation strategy that was successful in the past (cut a bit further to the right of the non-strategic cut). 

In EXP1, envy is generated in half or more of the cake divisions in all procedures for 3 and 4 agents with the exception of 3SC. In EXP1, envy is generated in over 90\% of the cases when 4EP is used. This finding is intriguing because 4EP theoretically performs well with regards to envy in that it minimizes the maximum number of players who can be envied among all proportional procedures \citep{Brams2011DivideandConquer}. More envy is generated in EXP1 than in EXP2, since in EXP1 all subjects are cutters, who are more prone to envy, whereas in EXP2 agents may play any role in the procedures. 

Both EXP1 and EXP2 support the hypothesis that envy-free procedures generate less envy. To see this, we compare the difference between the percentage of cases in which envy was generated in three-agent envy-free procedures (3SC) versus three-agent proportional procedures (3DS, 3LD). The difference is of 28 and 5 percentage points for EXP1 and EXP2, respectively, and is statistically significant in EXP1 and EXP2 when clustering at the subject level. To see this, we perform a cluster-adjusted t-test, with corresponding p-value of 0.001 and 0.0428 for EXP1 and EXP2, respectively. We note that the difference in envy generated is not statistically significant  when clustering at the session level; in this case the p-value becomes 0.1183. The increase in the p-value is likely due to having only 8 session-level clusters, rather than the existence of session-specific idiosyncratic effects.

Overall, envy decreases after subjects learn their opponents' preferences. To see this, we conduct a logit regression of the probability of the emergence of envy on the round number. The associated coefficients are -0.04 and -0.06 for EXP1 and EXP2, respectively, and are statistically significant (p-value equal to 0.001 and 0.002, respectively; see {\bf Table \ref{tab:regression2}}). 
\begin{table}[!h]
	\begin{center}
		\caption[caption,justification=centering]{Probability of envy explained by round number.}
		\label{tab:regression2}
		\begin{tabular}{lcccc}
			\toprule
	& \multicolumn{4}{c}{Dependent Variable:}    \\
& \multicolumn{4}{c}{Prob. of   envy}    \\
			& \multicolumn{2}{c}{EXP1} & \multicolumn{2}{c}{EXP2}      \\
			\midrule
			Round number   & -0.04     	& -0.04   		& -0.06     & -0.06\\
			Standard error & 0.011 		& 0.011   		& 0.018     & 0.019   \\
			Subject FE     & No          & Yes        & No    & Yes\\
			Session FE     &  -           &   -         & No     & Yes \\
			Observations   & 7,441       & 7,441      & 4,788  & 4,788\\
			P-value        & 0.000       & 0.000      & 0.001      & 0.002   \\
			Clusters       & 133         & 133        & 114     & 8  \\
			\bottomrule
		\end{tabular}
	
{\footnotesize	Logistic regression. Standard errors clustered at the subject level in the first three columns and at the session level in the fourth column. Removing subject fixed effects (FE) changes the results minimally.\hfil{ }}
	\end{center}
\end{table}

We summarize these findings as follows. 

\begin{result}
Envy-free procedures generate allocations with envy, but less than non-envy-free procedures. Envy decreases with learning.
\end{result}

\subsection{Learning}
\label{sub:results-learning}

We proceed to examining the extent to which knowledge can help with learning. We observe consistent evidence, in both EXP1 and EXP2, that knowing the opponents' preferences leads to higher payoffs in 2ACC, 3LD and 3SC (we compare the number of points obtained in rounds 1--5 versus those obtained in rounds 6--7, see {\bf Table \ref{tab:knowledge}}). In some cases, knowledge harms the subjects, such as in 2SCC and 2DS in EXP2 or in 3DS in EXP1, although the payoff difference is not statistically significant.
\begin{table}[t]
\begin{center}
\caption[caption,justification=centering]{Average points obtained with and without knowledge of opponents' preferences.}
\label{tab:knowledge}
\resizebox{\textwidth}{!}{
\begin{tabular}{l|ccccccccc}
\toprule
	& 2ACC 		& 2SCC & 2DS & 3DS & 3LD & 3SC & 4DS & 4EP & 4LD  \\
\midrule
	\multicolumn{10}{l}{\it EXP1}\\
	No knowl. \footnotesize ($n=665$)& 77   & 84   &$\cdot$	& 48  	& 40  	& 44  & 23  & 27  	& 41  	    \\

	Knowl.	\footnotesize ($n=266$)	& 87   & 88   & $\cdot$	& 45  	& 46  	& 49  & 25  & 28  	& 43  	    \\
	
	Difference 	& 10   & 4   &$\cdot$	& -3  	& 6  	& 6  & 2  	& 1  	& 2  	   \\
	\footnotesize p-value* & \footnotesize 0.00  & \footnotesize 0.07  & $\cdot$   & \footnotesize 0.20  & \footnotesize 0.01 & \footnotesize 0.00  & \footnotesize  0.18 & \footnotesize 0.27  & \footnotesize   0.31    \\

\midrule 
	\multicolumn{10}{l}{\it EXP2}\\
No knowl. \footnotesize ($n=570$)& 92  & 71  & 83 & 75 & 70 & 52&$\cdot$&$\cdot$&$\cdot$ \\

Knowl. \footnotesize($n=228$)	& 97  & 68 	& 81 & 81 & 78  & 54&$\cdot$&$\cdot$&$\cdot$	    \\

Difference 	& 5   & -3   	& -2  	& 6 &8&2  	   &$\cdot$&$\cdot$&$\cdot$\\
\footnotesize p-value* & \footnotesize  0.09 & \footnotesize 0.10   & \footnotesize 0.53  & \footnotesize 0.07  & \footnotesize 0.02  & \footnotesize 0.17  &$\cdot$&$\cdot$&$\cdot$\\
\footnotesize p-value** & \footnotesize  0.09 & \footnotesize 0.08   & \footnotesize  0.42 & \footnotesize 0.06  & \footnotesize  0.01 & \footnotesize 0.09  &$\cdot$&$\cdot$&$\cdot$\\
\bottomrule
\end{tabular}
}
\begin{tablenotes}
	\item \footnotesize We report the p-value for a cluster-adjusted t-test testing the null hypothesis that the difference is zero. One asterisk indicates standard errors clustered at the subject level, two indicate clustering at the session level. Rounding errors sometimes cause the difference to not match the original values exactly (up to 1 digit).
\end{tablenotes}
\end{center}
\end{table}

In EXP1, we find that most of the benefits of knowledge come from learning via experimentation, whereas in EXP2 agents learn from experimentation and directly observing their opponents' preferences jointly, but none of these two effects alone is significant on its own. {\bf Table \ref{tab:7}} shows that the payoffs obtained in round 5 compared to those in round 1 are significantly higher in EXP1 for 2ACC, 2SCC, $n$LD and 3SC, whereas we only observe this effect in EXP2 for 3LD. In comparison, in EXP1 revealing the opponents' preferences directly (round 7) only affects the payoffs obtained with the knowledge of experimentation (round 5) in 2ACC and 3SC, and does not affect the payoff in any procedure in EXP2.\footnote{We have interpreted the increase of payoffs in later rounds as the effects of learning agents' preferences. But, as one reviewer points out, the increase in payoffs could be also partially caused by agents learning about the division procedure. We believe that this second channel, while possible, is relatively small: we explained the procedures in detail to the participants before the experiment started and, in EXP2, agents do a practice round to familiarize themselves with the division mechanisms. Nonetheless, we acknowledge that it would be interesting to precisely disentangle the consequences of learning about preferences versus mechanisms in future experiments.}
\begin{table}[h!]
	\begin{center}
	\caption[caption,justification=centering]{Average points obtained in rounds 1, 5 and 7.}
	\label{tab:7}
	\resizebox{\textwidth}{!}{
		\begin{tabular}{l|ccccccccc}
			\toprule
			& 2ACC 		& 2SCC & 2DS & 3DS & 3LD & 3SC & 4DS & 4EP & 4LD  \\
			\midrule
			\multicolumn{10}{l}{\it EXP1}\\
			 Round 1  \footnotesize ($n=133$)  & 68     & 81&  $\cdot$   & 50     & 34     & 40     & 23     & 34     & 25     \\
			 Round 5  \footnotesize ($n=133$)  & 82     & 91 &  $\cdot$  & 50     & 45     & 44     & 24     & 45     & 28     \\
			 Round 7  \footnotesize $(n=133)$  & 89   & 91  &$\cdot$ & 46   & 47   & 50   & 26   & 46   & 29   \\
			Diff rounds 5 - 1 & 14     & 10  & $\cdot$  & 0      & 11     & 5      & 1      & 8      & 3      \\
			\footnotesize	p-value*    & \footnotesize 0.00	& \footnotesize 0.00	 &$\cdot$ & \footnotesize 0.97	  & \footnotesize 0.00   & \footnotesize 0.05	& \footnotesize 0.66	 & \footnotesize 0.01	  & \footnotesize 0.07\\
		Diff rounds 7 - 5 & 7    & 0    & $\cdot$& -4   & 2    & 6    & 2    & 0    & 1    \\
		\footnotesize	p-value*    & \footnotesize 0.04 & \footnotesize 0.89 &$\cdot$ &\footnotesize 0.20 & \footnotesize 0.51 & \footnotesize 0.02 & \footnotesize 0.30 & \footnotesize 0.96 & \footnotesize 0.58 \\
			\midrule 
			\multicolumn{10}{l}{\it EXP2}\\
			Round 1 \footnotesize ($n=114$) & 90  & 73  & 85 & 72 & 60 & 51&$\cdot$&$\cdot$&$\cdot$ \\
		Round 5 \footnotesize ($n=114$)	& 94  & 71 	& 83 & 76 & 77  & 51	 &$\cdot$&$\cdot$&$\cdot$  \\
		Round 7 \footnotesize ($n=114$)& 99  & 68  & 84 & 79 & 81 & 53 &$\cdot$&$\cdot$&$\cdot$\\
		Diff rounds 5 - 1	& 4   & -2   	& -3  	& 5 & 17 &1  	 &$\cdot$&$\cdot$&$\cdot$  \\
		\footnotesize p-value* & \footnotesize 0.30   & \footnotesize 0.52   & \footnotesize 0.45  & \footnotesize 0.29  & \footnotesize 0.00  & \footnotesize 0.78  &$\cdot$&$\cdot$&$\cdot$\\
		\footnotesize p-value** & \footnotesize 0.36  & \footnotesize  0.53  & \footnotesize 0.46  & \footnotesize 0.32  & \footnotesize  0.00 & \footnotesize 0.78  &$\cdot$&$\cdot$&$\cdot$\\
					Diff rounds 7 - 5 	& 5   & -3   	& 1  	& 3 & 4 & 2 &$\cdot$&$\cdot$&$\cdot$  	   \\
				\footnotesize p-value* & \footnotesize 0.09   & \footnotesize 0.29   & \footnotesize 0.72  & \footnotesize 0.52  & \footnotesize 0.29  & \footnotesize 0.25  &$\cdot$&$\cdot$&$\cdot$\\
				\footnotesize p-value** & \footnotesize  0.27 & \footnotesize  0.30  & \footnotesize 0.72  & \footnotesize 0.54  & \footnotesize 0.32  & \footnotesize  0.27 &$\cdot$&$\cdot$&$\cdot$\\
			\bottomrule
		\end{tabular}
	}
		\begin{tablenotes}
			\item \footnotesize We report the p-value for a cluster-adjusted t-test testing the null hypothesis that the difference is zero. One asterisk indicates standard errors clustered at the subject level (no effect, since there is only one observation per subject in each round), two indicate clustering at the session level. Rounding errors sometimes cause the difference to not match the original values exactly (up to 1 digit).
		\end{tablenotes}
	\end{center}
\end{table}

\subsection{Quality of Learning}
\label{sub:results-quality}
In this section we are interested in observing how good people are at learning their opponent’s valuations.
To investigate this question, we focus on 2ACC, a procedure in which any opponent who is the chooser (either a non-strategic automaton or a strategic agent) always reports her real preferences by choosing her preferred part of the cake among the two available ones.\footnote{Note that this analysis does not assume that the agent's opponent in 2ACC is truthful. The only necessary assumption is that the player is maximizing their own utility when choosing the best piece of the cake out of the two available.}
We present a model of rational learning for the cutter in 2ACC. We consider 2ACC played for $T$ rounds, where in each round, Alice cuts the cake and Bob chooses a piece. Our model follows the experiment setup, particularly:
\begin{itemize}
\item The game is discretized: the cake is $[0,c]$ for some integer $c$ (in the experiment $c=600$ is the number of pixels in the cake); Alice may cut only in integer locations; a cut in $x$ means that the left piece is $[0,x)$ and the right piece is $[x,c]$.
\item Bob always picks the most valuable piece for him, and if the pieces have equal value, he breaks the tie by selecting the left piece.
\end{itemize}
Alice's payoff depends only on Bob's \emph{half point} --- the integer $h$ for which $v_B(0,h) = v_B(h,c) = v_B(0,c)/2$.
If Alice cuts at some $x< h$, then Bob takes the right piece and she gets $[0,x)$; if Alice cuts at $x\geq h$, then Bob takes the left piece and she gets $[x,c]$.

If Alice knows $h$, then it is optimal for her to cut either at $h-1$ or at $h$; in the former case she gets $[0,h-1)$ and in the latter case she gets $[h,c]$. Therefore Alice can guarantee to herself a utility of:
\begin{align*}
u^{opt}(h) =
\max[v_A(0,h-1), v_A(h,c)]
\end{align*}
Initially, Alice does not know $h$, but she can learn the possible range of $h$ from Bob's choices: if Alice cuts at some $s\in[0,c]$ and Bob chooses the right piece, she learns that $h> s$; similarly, if Alice cuts at $t\in[0,c]$ and Bob chooses the left piece, Alice learns that $h \leq t$. In each round,
Alice's knowledge about Bob is summarized by two numbers $s<t$ that represent the lower and upper bounds for Bob's half-point $h$, i.e., $s < h \leq t$.
With this knowledge, cutting at any $x<s$ is dominated by cutting at $s$ (since Alice will get $[0,x]$, which is worth at most what $[0,s]$ is worth to her), and cutting at any $x > t$ is dominated by cutting at $t$ (since Alice will get $[x+1,c]$, which is worth at most what $[t+1,c]$ is worth to her).
We say that Alice is \emph{rational} if all her cuts (from the second round onwards) are undominated.

Our findings regarding rational agents and the use of undominated strategies are summarized in {\bf Table \ref{tab:rationality2}}. Interestingly, less than 30\% of all players are fully rational (i.e. all their actions are undominated). 
Moreover, even a relaxed definition of rationality (that we call ``semi-rationality''), that allows for one mistake, is satisfied by only one-third to one-half of our subjects. 
While our finding, that there is only a minority of rational or semi-rational behavior, is in line with previous studies showing that human subjects often play dominated strategies
\citep{artemov2017,hassidim2016,hassidim2017mechanism,rees2017,parco2004enhancing}, we find it somewhat surprising that so many people behave in a way that is so clearly irrational.

\begin{table}[t]
\begin{center}
\caption{
	\label{tab:rationality2}
	Percentage of (semi)-rational behavior in 2ACC.
}
\begin{tabular}{l|cc}
\toprule
& \shortstack{Rational players\\(no dominated actions)}
& \shortstack{Semi-rational players\\($\leq 1$ dominated actions)}
\\
\midrule
2ACC EXP1 
& 32/132 = 24.2\% 
& 50/132 = 37.9\%
\\
2ACC EXP2 
& 17/57 = 29.8\% 
& 30/57 = 52.6\% 
\\
\bottomrule
\end{tabular}
\end{center}
\end{table}

\section{Conclusion}
\label{sec:conclusion}
We conduct two lab experiments involving several well-known fair cake cutting procedures in an attempt to quantify the extent to which stronger theoretical (non-strategic) fairness properties correspond to ``better'' performance in practice. In particular, we consider six proportional procedures, three of which are also envy-free.  Since the envy-freeness property  can only be guaranteed when agents do not manipulate the cake-cutting procedures, it is not clear if envy-free procedures will lead to reduced envy in the lab, where subjects very often report their preferences strategically. This work contrasts the level of envy generated in the outcomes of the Asymmetric cut-and-choose, Symmetric cut-and-choose, and Selfridge-Conway procedures, which are all envy free, to the envy generated in the outcomes of the Dubins-Spanier, Last diminisher and Even-Paz procedures, which are simply proportional.


Our experiment provides the first empirical evidence supporting the real-life application of the celebrated Selfridge-Conway cake-cutting procedure, among other interesting observations. Our experiment strongly suggests that the Selfridge-Conway procedure generates less envy than other proportional procedures. We hope that our findings guide its practical implementation, in the light of the very successful implementations of other fair division protocols in online platforms such as Spliddit.com.

Three interesting directions for future experiments are: 
(a) Check other cake-cutting procedures, in particular, procedures that guarantee additional properties such as equitability, strategy-proofness or Pareto-efficiency. Is the added complexity of these procedures justified? 
(b) Compare the performance of structured cake-cutting procedures to unstructured face-to-face bargaining.
(c) Check division of more realistic resources. For example, instead of showing the subjects artificial one-dimensional ``cakes'', one can show them real two-dimensional maps of land-estates. Fair division of land is an important issue in many inheritance and dissolution cases. How can cake-cutting procedures be used to solve such issues in practice?

\clearpage
\setlength{\bibsep}{0cm}
\bibliographystyle{ecta}
\bibliography{biblio}

\clearpage
\begin{appendices}

\section*{Appendix 1: Preference Profiles}

All the preference profiles are generated using piecewise uniform valuations. The cake is divided in 600 pixels of equal length with each desired pixel giving the agent 1 point. Agents desire 120 pixels which give the corresponding 120 points described in the main text.
We present the preferences using the tables below; a one in the table indicates that the agent desires the interval in question. The intervals that are not mentioned are not desired by any agent.

\begin{table}[!ht]
\begin{center}
\caption[caption,justification=centering]{Preferences used in 2ACC.}
\label{tab:2acc}				
\begin{tabular}{|l|l|l|l|l|l|l|}
\toprule
2ACC    & 61-120 & 121-130 & 171-190 & 291-310 & 411-430 & 451-540 \\
\hline
Subject 1 & 1      & 0       & 1       & 1       & 1       & 0       \\
Subject 2 & 0      & 1       & 0       & 0       & 1       & 1      \\
\bottomrule
\end{tabular}
\end{center}
\end{table}

\begin{table}[!ht]
\begin{center}
\caption[caption,justification=centering]{Preferences used in 2SCC.}
\label{tab:2scc}				
\scalebox{.7}{\begin{tabular}{|l|l|l|l|l|l|l|l|l|l|l|}
\toprule
2SCC    & 141-170 & 191-220 & 231-240 & 241-260 & 271-300 & 311-320 & 321-330 & 361-390 & 471-490 & 511-540 \\
\hline
Subject 1 & 0       & 0       & 1       & 1       & 1       & 1       & 0       & 0       & 1       & 1       \\
Subject 2 & 1       & 1       & 1       & 0       & 0       & 1       & 1       & 1       & 0       & 0      \\
\bottomrule
\end{tabular}}	\end{center}
\end{table}

\begin{table}[!ht]
	\begin{center}
		\caption[caption,justification=centering]{Preferences used in 2DS.}
		\label{tab:2ds}				
		\scalebox{.7}{
			\begin{tabular}{|l|l|l|l|l|l|l|l|l|l|l|}
				\toprule
				2DS&	101-130 & 151-180 & 211-240 & 241-260 & 271-280 & 291-320 & 321-340 & 341-350 & 351-370 \\
				\hline
			Subject 1 &	0       & 0       & 1       & 1       & 1       & 1       & 0       & 1       & 1       \\
				Subject 2&1       & 1       & 1       & 0       & 0       & 0       & 1       & 1       & 0      \\
				\bottomrule
		\end{tabular}
	}
	\end{center}
\end{table}

\begin{table}[!ht]
\begin{center}
\caption[caption,justification=centering]{Preferences used in 3DS.}
\label{tab:3ds}				
\scalebox{.7}{\begin{tabular}{|l|l|l|l|l|l|l|l|l|l|l|}
\toprule
3DS     & 71-110 & 121-130 & 131-150 & 151-160 & 171-180 & 191-200 & 271-310 & 311-380 & 411-430 & 451-540 \\
\hline
Subject 1& 1      & 1       & 1       & 1       & 0       & 0       & 1       & 0       & 0       & 0       \\
Subject 2& 0      & 1       & 0       & 0       & 0       & 0       & 0       & 0       & 1       & 1       \\
Subject 3& 0      & 1       & 1       & 0       & 1       & 1       & 0       & 1       & 0       & 0      \\
\bottomrule
\end{tabular}}
\end{center}
\end{table}

\begin{table}[!ht]
	\begin{center}
		\caption[caption,justification=centering]{Preferences used in 3LD.}
		\label{tab:3ld}				
		\scalebox{0.53}{	
			\begin{tabular}{|l|l|l|l|l|l|l|l|l|l|l|l|l|l|}
				\toprule
				3LD     & 71-90 & 91-110 & 121-190 & 221-230 & 231-260 & 281-300 & 301-320 & 341-350 & 351-370 & 371-400 & 401-410 & 431-440 & 451-460 \\
				\hline
				Subject 1& 0     & 1      & 0       & 1       & 1       & 1       & 0       & 1       & 1       & 0       & 1       & 0       & 0       \\
Subject 2& 1     & 1      & 1       & 1       & 0       & 0       & 0       & 0       & 0       & 0       & 0       & 0       & 0       \\
		Subject 3 & 0     & 0      & 0       & 0       & 0       & 1       & 1       & 0       & 1       & 1       & 1       & 1       & 1     \\
			\bottomrule
		\end{tabular}}
	\end{center}
\end{table}

\begin{table}[!ht]
	\begin{center}
		\caption[caption,justification=centering]{Preferences used in 3SC.}
		\label{tab:3sc}				
		\scalebox{0.42	}{	
			\begin{tabular}{|l|l|l|l|l|l|l|l|l|l|l|l|l|l|l|l|l|l|}
			\toprule
				3SC     & 71-80 & 81-90 & 91-100 & 101-110 & 141-150 & 151-170 & 171-190 & 211-230 & 271-280 & 281-290 & 291-300 & 301-320 & 321-330 & 331-340 & 381-400 & 451-470 & 471-490 \\
				\hline
				Subject 1& 0     & 0     & 0      & 0       & 0       & 1       & 1       & 1       & 0       & 0       & 0       & 0       & 0       & 0       & 1       & 1       & 1       \\
				Subject 2& 0     & 1     & 0      & 1       & 1       & 1       & 0       & 0       & 1       & 1       & 0       & 1       & 0       & 1       & 0       & 0       & 1       \\
				Subject 3 & 1     & 1     & 1      & 1       & 0       & 0       & 0       & 0       & 0       & 1       & 1       & 1       & 1       & 1       & 0       & 0       & 1      \\
			\bottomrule		\end{tabular}	
		}
	\end{center}
\end{table}

\begin{table}[!ht]
\begin{center}
\caption[caption,justification=centering]{Preferences used in 4DS.}
\label{tab:4ds}				
\scalebox{0.415}{	\begin{tabular}{|l|l|l|l|l|l|l|l|l|l|l|l|l|l|l|l|l|l|l|}
\toprule
4DS     & 61-80 & 81-90 & 91-120 & 141-150 & 151-170 & 171-180 & 181-210 & 211-240 & 241-270 & 271-300 & 301-330 & 331-360 & 371-390 & 391-420 & 421-450 & 451-480 & 491-510 & 511-540 \\
\hline
Subject 1& 1     & 0     & 0      & 1       & 1       & 0       & 0       & 0       & 1       & 0       & 0       & 0       & 1       & 0       & 0       & 0       & 1       & 0       \\
Subject 2& 1     & 1     & 0      & 0       & 0       & 0       & 1       & 0       & 0       & 0       & 1       & 0       & 0       & 0       & 1       & 0       & 0       & 0       \\
Subject 3& 0     & 0     & 1      & 0       & 0       & 0       & 0       & 1       & 0       & 0       & 0       & 1       & 0       & 0       & 0       & 1       & 0       & 0       \\
Subject 4& 0     & 0     & 0      & 0       & 1       & 1       & 0       & 0       & 0       & 1       & 0       & 0       & 0       & 1       & 0       & 0       & 0       & 1      \\
\bottomrule
\end{tabular}}
\end{center}
\end{table}

\begin{table}[!ht]
\begin{center}
\caption[caption,justification=centering]{Preferences used in 4LD.}
\label{tab:4ld}				
\scalebox{0.5}{	
\begin{tabular}{|l|l|l|l|l|l|l|l|l|l|l|l|l|l|l|}
\toprule
4LD     & 61-90 & 91-110 & 111-160 & 181-230 & 231-250 & 251-270 & 271-280 & 281-290 & 311-340 & 341-350 & 351-370 & 371-380 & 381-410 & 421-520 \\
\hline
Subject 1& 1     & 0      & 0       & 1       & 1       & 0       & 1       & 1       & 0       & 0       & 0       & 0       & 0       & 0       \\
Subject 2& 0     & 0      & 1       & 0       & 0       & 1       & 1       & 0       & 0       & 1       & 1       & 1       & 0       & 0       \\
Subject 3& 0     & 0      & 0       & 0       & 1       & 1       & 0       & 0       & 1       & 1       & 0       & 1       & 1       & 0       \\
Subject 4& 0     & 1      & 0       & 0       & 0       & 0       & 0       & 0       & 0       & 0       & 0       & 0       & 0       & 1      \\
\bottomrule
\end{tabular}
}
\end{center}
\end{table}

\begin{table}[!ht]
\begin{center}
\caption[caption,justification=centering]{Preferences used in 4EP.}
\label{tab:4ep}				
\scalebox{0.42}{	
\begin{tabular}{|l|l|l|l|l|l|l|l|l|l|l|l|l|l|l|l|l|l|l|}
\toprule
4EP     & 91-110 & 111-120 & 121-140 & 161-170 & 171-190 & 191-210 & 211-220 & 221-240 & 241-270 & 281-300 & 301-320 & 331-340 & 341-350 & 351-360 & 361-370 & 411-430 & 471-510 \\
\hline
Subject & 1      & 1       & 0       & 0       & 1       & 1       & 1       & 1       & 0       & 0       & 1       & 0       & 0       & 0       & 0       & 0       & 0       \\
Robot 1 & 0      & 1       & 1       & 0       & 1       & 1       & 0       & 0       & 0       & 1       & 1       & 0       & 0       & 1       & 0       & 0       & 0       \\
Robot 2 & 0      & 0       & 0       & 0       & 0       & 1       & 1       & 0       & 0       & 1       & 1       & 0       & 1       & 1       & 1       & 1       & 0       \\
Robot 3 & 0      & 0       & 0       & 1       & 1       & 0       & 0       & 0       & 1       & 0       & 0       & 1       & 1       & 0       & 0       & 0       & 1      \\
\bottomrule
\end{tabular}}
\end{center}
\end{table}

\clearpage
\section*{Appendix 2: Experiment Instructions}

Upon their arrival to the lab, the cake-cutting procedures are explained to the subjects using the slides available at \href{www.josueortega.com}{www.josueortega.com}. We do not include them here for the sake of brevity. The presentation comprises 31 slides so to make the procedures as clear as possible. The instructions for EXP1 and EXP2 are almost identical. The only differences are that, in EXP2, we changed the name of ``{\it Super Fair} to ``{\it Double Knife}'', and that subjects are told that their opponents will be real persons who are also in the room. EXP2 includes {\it ``Leftmost Leaves''} for two players, and excludes all 4-agent procedures. In EXP1, the procedures are shown in a fixed order, whereas in EXP2 the procedures are shown in a random order.

We describe the text that the subjects observe in the graphical interface. These are as follows:\\

Welcome to the game. When you are ready to start click the start button.

\paragraph{I Cut You choose, against 1 opponent} Description: You will cut the cake into two parts. Your opponent will choose the one he prefers. You will receive the other one. Suggestion: If you cut the cake in two pieces worth 60 points, you guarantee that you will receive 60 points. Dividing the cake differently may give you more points, but may also give you less.

\paragraph{Cut Middle, against 1 opponent} Description: You will cut the cake into two parts. Your opponent also cuts the cake into two. We cut the cake in the middle of those cuts and you get the part that includes your cut. Suggestion: If you cut the cake in two pieces worth 60 points, you guarantee that you will receive at least 60 points. Dividing the cake differently may give you more points, but may also give you less.

\paragraph{Leftmost Leaves, against 1 opponent} Description: All players make one cut to the cake. The one who cuts the leftmost piece gets the left part, and the other player gets the right part. Suggestion: If you cut the cake at a point which makes the left piece to have a value of 60, you guarantee that you will receive at least 60 points in this round. Dividing the cake differently may give you more points, but may also give you less.

\paragraph{Leftmost Leaves, against 2 opponents} Description: All players make one cut to the cake. The one who cuts the leftmost piece gets that part and leaves. The procedure is repeated until no agent is left. You may need to cut the cake twice in the same round if you don't choose the leftmost piece right away. Suggestion: If you cut the cake at 40 in each stage, you guarantee at least 40 points. Dividing the cake differently may give you more points, but may also give you less.

\paragraph{Leftmost Leaves, against 3 opponents} Description: All players make one cut to the cake. The one who cuts the leftmost piece gets that part and leaves. The procedure is repeated until no agent is left. You may need to cut the cake twice in the same round if you don't choose the leftmost piece right away. Suggestion: If you cut the cake at 30 in each stage, you guarantee at least 30 points. Dividing the cake differently may give you more points, but may also give you less.

\paragraph{Last Challenger, against 2 opponents} Description: You make a cut to the cake. This cut can be challenged by other players. If it is not challenged, you get the left piece of the cake and leave. If it is challenged, the player who challenges gets the left piece and leaves, and we restart the procedure with the leftover cake. You may need to cut the cake twice in the same round if your initial cut is challenged. Suggestion: If you cut the cake at 40 in each stage, you guarantee at least 40 points. Dividing the cake differently may give you more points, but may also give you less.

\paragraph{Last Challenger, against 3 opponents} Description: You make a cut to the cake. This cut can be challenged by other players. If it is not challenged, you get the left piece of the cake and leave. If it is challenged, the player who challenges gets the left piece and leaves, and we restart the procedure with the leftover cake. You may need to cut the cake twice in the same round if your initial cut is challenged. Suggestion: If you cut the cake at 30 in each stage, you guarantee at least 30 points. Dividing the cake differently may give you more points, but may also give you less.

\paragraph{Super Fast, against 3 opponents} Description: All players split the cake into two. The two who choose the leftmost cuts divide the first half, the other two the second half. Each half is divided using leftmost leaves. You will have to cut the cake twice. Suggestion: If you first cut the cake at 60 points and then at 30, you guarantee at least 30 points. Dividing the cake differently may give you more points, but may also give you less.

\paragraph{Super Fair$\slash$Double Knife, against 2 opponents} Description: In this procedure you have two knives. You should cut the cake into three pieces. Then a complex procedure occurs, which you can read in your information sheet. Suggestion: If you cut the cake into three pieces worth 40 points each, you guarantee 40 points. Dividing the cake differently may give you more points, but may also give you less.

\paragraph{Additional Explanation for 3SC} You will cut the cake into three pieces using two knives. We suggest you to cut the cake into three pieces worth 40 points each so to guarantee yourself 40 points. Dividing the cake differently may give you more points, but may also give you less.

After you cut the cake, opponent 1 will trim her most valued piece so to make her two most preferred pieces of equal value. The part she cuts from her most valued piece of cake will be put apart and divided later (the trimmings). Then opponent 2 will take the part he prefers. If opponent 2 does not take the part that opponent 1 trims, then opponent 1 will receive that part and you will receive the leftover. Otherwise, in case opponent 2 picks the trimmed part, opponent 1 chooses one of the two remaining pieces and then you choose last.

Once the main pieces of the cake have been divided, we will divide the trimmings. One of the two opponents (the one who did not choose the trimmed part) will cut the trimmings into three pieces. Then the other opponent will choose one of them. From the two leftovers, you will be given the one which is best for you, and the last one will be given to the remaining opponent.\\

Subjects also receive an official information sheet with the following information:

Strategic Behavior in Fair Division Problems

\paragraph{Invitation to our study} We would like to invite you to participate in this research project. You should only participate if you want to; choosing not to take part will not disadvantage you in any way. Before you decide whether you want to take part, it is important for you to read the following information carefully and discuss it with others if you wish. Ask us if there is anything that is not clear or you would like more information.

\paragraph{Background on the project} We are conducting an exploration of how people make economic decisions, in particular on how they decide to divide and share resources with others. We are testing how different resource allocation methods affect the economic decisions people make.

\paragraph{Experiment} You will be asked to divide resources with 2, 3, or 4 other agents. The way in which you decide to divide the resources will affect how much money you will receive by the end of the experiment. The experiment will last for around one hour. You won’t be required to participate again in the experiment. You will be paid in private at the end of the experiment. You will receive at least $\pounds$5 for showing up, but you may earn more money based on your decisions throughout this session.

\paragraph{Are there any risks associated with this experiment?} There are no risks associated with this experiment. Shall you experience any discomfort please contact any member of the staff.

\paragraph{Informed consent} Should you agree to take part in this experiment, you will be asked to sign a consent form before the experiment commences.

\paragraph{Withdrawal} Your participation is voluntary and you will be free to withdraw from the project at any time without giving any reason and without penalty. If you wish to withdraw, you simply need to notify the principal investigator (see contact details below). If any data have already been collected, upon withdrawal, your data will be destroyed, unless you inform the principal investigator that you are happy for us to use such data for the scientific purposes of the project.

\paragraph{Data gathered} We will record the economic decisions you make during the experiment, namely how you decide to share resources with other participants. Signed consent forms will be kept separately from individual experimental data and locked in a drawer until the end of the project.

\paragraph{Findings} After the end of the project, we will publish the findings of our research. We will be happy to provide you with a lay summary of the main findings and with copies of the articles published if you express an interest.

\paragraph{Concerns and complaints} If you have any concerns about any aspect of the study or you have a complaint, in the first instance please contact the principal investigator of the project (see contact details below). If are still concerned or you think your complaint has not been addressed to your satisfaction, please contact the Director of Research in the principal investigator’s department (see below). If you are still not satisfied, please contact the University’s Research Governance and Planning Manager (Sarah Manning-Press).

\paragraph{Funding} The research is funded by the EssexLab of the University of Essex.

\paragraph{Ethical approval} This project has been reviewed on behalf of the University of Essex Ethics Committee and had been given approval.

\paragraph{Principal investigator} Dr. Josue Ortega, Lecturer, Department of Economics, University of Essex, Wivenhoe Park, CO4 3SQ, Colchester,

\noindent josue.ortega@essex.ac.uk.

\paragraph{Co-investigators} Dr. Maria Kyropoulou, Lecturer, Department of Computer Science and Electronic Engineering, University of Essex
Wivenhoe Park, CO4 3SQ, Colchester, maria.kyropoulou@essex.ac.uk.

\noindent Dr. Erel Segal-Halevi, Lecturer, Department of Computer Science, Ariel University, Ramat HaGolan St 65, Ari'el, erelsgl@gmail.com

\paragraph{Director of Research, Economics Department} Prof. Friederike Mengel, Professor, Department of Economics, University of Essex
Wivenhoe Park, CO4 3SQ, Colchester, fmengel@essex.ac.uk.

\paragraph{Research Governance and Planning Manager} Sarah Manning-Press, University of Essex, Wivenhoe Park, CO4 3SQ, Colchester, sarahm@essex.ac.uk.\\

Finally, we include the questions in the fairness survey that subjects complete after they finish cutting all the cakes. The observations corresponding to this survey were not further analyzed because of potential experimenter-demand effects.

\paragraph{Experiment feedback} Please answer (with as many details as possible) the following questions.\\

\noindent How fair was ``Cut and choose"? {\it Very unfair, Unfair, Fair, Very fair.}

\noindent Feedback: {\it textbox}.\\

\noindent How fair was ``middle cut"? {\it Very unfair, Unfair, Fair, Very fair.}

\noindent Feedback: {\it textbox}.\\

\noindent How fair was ``last challenger"? {\it Very unfair, Unfair, Fair, Very fair.}

\noindent Feedback: {\it textbox}.\\

\noindent How fair was ``lefmost leaves"? {\it Very unfair, Unfair, Fair, Very fair.}

\noindent Feedback: {\it textbox}.\\

\noindent In your opinion, was ``super fair" a fairer procedure than all the others? {\it Yes, No.}

\noindent Feedback: {\it textbox}.\\

\noindent In your opinion, was ``super fast" an easier procedure to use than all the others? {\it Yes, No.}

\noindent Feedback: {\it textbox}.\\

\noindent Would you have preferred to bargain over the cake directly with the other players instead of dividing it with these methods? {\it Yes, No, Doesn't matter.}\\

\noindent Please give us your comments on which procedures produced fairer allocations and were easier to use.

\noindent Feedback: {\it textbox}.\\

\pagebreak

\section*{Appendix 3: Omitted Proofs}

We present the proofs omitted in the main text. 

\begin{proposition*}[\ref{lem:max-strategising-payoff}]
The procedures 2ACC, 2SCC, $n$DS, $n$LD, $n$EP, 3SC are $\left(1-\frac{1}{n}\right)$-strategy-proof and this is tight.
\end{proposition*}

\begin{proof}[Proof of Proposition \ref{lem:max-strategising-payoff}]
The fact that all these procedures are porportional, implies that if agents adhere to the procedure, then each of them is guaranteed utility $1/n$. Since they can get utility at most $1$ in any allocation, the increase in their utility by strategic behavior is at most $1-1/n$. To prove that this is tight, we provide instances such that an agent would get utility exactly $1/n$ by non-strategically reporting her valuation function, while she could get utility $1$ by strategizing.

We start with the case of 3SC. Consider a cake $[0,1]$ and the following valuations of the agents: $v_1(0,1/3)=v_2(1/3,2/3)=v_3(2/3,1)=1$; agents have valuation $0$ for any other part. Assume everyone behaves non-strategically, and in the first step agent 1 divides the cake in the following parts of equal value to her: $[0,1/9), [1/9,2/9)$, and $[2/9,1]$. Agent 2 has positive valuation only for the last part, so in the next step, she will trim it so that the trimmed part has value $0$ to her; let the trimmed part lie inside $[2/9,1/3]$. 
Agent 3 is indifferent between the pieces, so let her choose the leftmost one. Agent 2 will then get the trimmed piece, and the trimmings will be split by agent 3 such that both agents 2 and 1 only have positive value for the leftmost part of the trimmings. Since agent 2 selects first, agent 1's overall utility will be exactly $1/3$.

Now imagine that agent 1 behaves strategically in the first step and divides the cake into the parts $[0,1/3), [1/3,2/3)$, and $[2/3,1]$. Agent 2 will trim the second part so that the trimmed part is negligible, i.e. it is worth $0$ to everyone. Agent 3 will rationally get her desired part, i.e. $[2/3,1]$, agent 2 will get the trimmed part, and agent 1 will get her desired part $[0,1/3)$, thus obtaining utility $1$. 

Regarding 2SCC, consider a cake $[0,1]$ and the following valuations of the agents, for some positive $\epsilon<1/16$: $v_1(0,1/4)=v_1(3/8+\epsilon,1/2-\epsilon)=\frac{1}{2}$, and $v_2(1/2,1/2+\epsilon)=v_2(1/2+\epsilon,1)=\frac{1}{2}$; agents have valuation $0$ for any other part. Assume everyone behaves non-strategically, and agents 1 and 2 cut the cake at points $1/4$ and $1/2+\epsilon$, respectively\footnote{For the given instance, agent 1 would not violate the protocol by cutting within $(1/4,3/8-\epsilon)$, instead of at $1/4$, however, we can eliminate this ambiguity by slightly modifying the instance. In particular, let agent 1 have negligibly small positive utility for part $[1/4,1/4+\epsilon)$, and correspondingly decrease her utility for her rightmost desired item. Then, agent 1 will necessarily cut at $1/4$ if she is not strategic.}, to divide it to two parts of equal value to them. After the end of the procedure agent 1 will receive utility $1/2$. Now imagine that agent 1 behaves strategically and cuts the cake at $1/2-\epsilon$. In the resulting allocation each of the agents will receive utility $1$.

The other cases are simpler and use instances with valuation functions of the form $v_i(\frac{i-1}{n},\frac{i}{n})=1$ for $i=1,\ldots, n$, and $0$ otherwise, similar to 3SC. The analysis is straightforward (similar, yet much simpler than the one for 3SC), hence we omit it.\end{proof}

\begin{proposition*}[\ref{lem:envy-when-strategizing}]
Envy can be generated in 3SC with just one agent misrepresenting her preferences. This agent achieves a higher payoff at the cost of being envious.
\end{proposition*}

\begin{proof}[Proof of Proposition \ref{lem:envy-when-strategizing}]
We present an instance and a corresponding strategy for agent 1 who is assumed to be strategic and tries to maximize her utility when competing with two non-strategic agents. We show that agent 1 will end up envious of another agent, although she will achieve higher utility than what she would get by behaving non-strategically. We focus on the action of agent 1 at the beginning of the process, when she is asked to split the cake into three pieces. We consider this to be the strategy of agent 1; w.l.o.g. we ignore subsequent actions in the analysis as the only other choice that agent 1 makes is to select a part of the trimmings close to the end of the process, and it is clear that her incentives at that point are aligned with behaving non-strategically and getting the part that is most valuable to her.

Consider a cake $[0,1]$, which comprises 6 parts. The preferences of the agents are described by the valuations in Table \ref{tab:envyatNE}; agents are assumed to have uniform valuations within each of these parts.

\begin{table}[!ht]
\begin{center}
	\caption[caption,justification=centering]{Agents' preferences over cake pieces such that agent 1's optimal strategy in 3SC makes her envious.}
	\label{tab:envyatNE}				
	\begin{tabular}{lcccccc}
	\toprule
	Agents $\slash$ Cake parts &$P_1$   &$P_2$ &$P_3$  & $P_4$& $P_5$& $P_6$\\
	\midrule
	&&&&&&\\
	Agent 1 & 0     & $1/3$     & $0$       & $1/3$       & $0$      & $1/3$        \\
	&&&&&&\\
	Agents 2 and 3 & $1/6 $   & $1/6  $     & $1/3$     & 0       & $1/3$      & 0     \\
	&&&&&&\\
	\bottomrule
\end{tabular}
\end{center}
\end{table}

For consistency, we will make the following assumption regarding the behavior of the non-strategic agents 2 and 3. We assume that among actions that result in the same utility, the non-strategic agents will choose the one that immediately harms agent 1 the most. If there still is a tie, the agents will chose the leftmost valid option. We also assume that agent 1 cannot cut within the parts for which she has utility $0$, i.e. $P_1, P_3,$ and $P_5$ (the instance could be defined so that these parts have a negligibly small width and the space of allowed cuts is discrete). For the smooth execution of the protocol we allow such cuts if and only if it is absolutely necessary in order to achieve an exact trimmed piece or an even distribution of trimmings.

Non-strategic behavior for agent 1 would imply that she divides the cake at three equally valued pieces, i.e. $P_1\cup P_2$,   $P_3\cup P_4$, and $P_5\cup P_6$. This split would result in utility $1/3$ for agent 1 as there would essentially be no trimming and each of the other agents would obtain one of these pieces. 

We now show that the split in pieces $P_1\cup P_2\cup P_3\cup P_4$, $P_5$, and $P_6$, is a better response for agent 1 than behaving non-strategically, yet makes her envious of another agent. Indeed, under this split, agent 2 will trim the first piece so that the trimmings $T=P_1\cup P_2$. Agents 2 and 3 will first choose piece $P_3\cup P_4$ (the trimmed piece), and $P_5$, respectively, leaving part $P_6$ for agent 1. 
The trimmings  will be split in three equal parts, for which agent 1 will have value 0, 1/9, and 2/9, respectively. Whoever got the trimmed piece, $P_3\cup P_4$, will first select the rightmost part of the trimmings by our assumption, leaving the middle one for agent 1. Overall, agent 1's allocation has value $\frac{1}{3}+\frac{1}{9}=\frac{4}{9}>\frac{1}{3},$ yet she is envious of the agent who got the trimmed part, as her allocated piece has total value $\frac{1}{3}+\frac{2}{9}=\frac{5}{9}$ for agent 1.

It remains to show that the split in pieces $P_1\cup P_2\cup P_3\cup P_4$, $P_5$, and $P_6$ is a best response strategy for agent 1 assuming that the other two agents behave non-strategically.
Indeed, we show that $4/9$ is the maximum utility she can get, by examining all other possible cuts she could make at the beginning of the procedure. Let $c_1$ and $c_2$ denote the cuts of agent 1 and let  $c_i\in P$, for $i=1,2$, denote the fact that the $i$-th cut is inside part $P$ (including its boundary). 
\begin{itemize}
\item $c_1,c_2\in P_6$. We get that $T=(0,c_1)$ and agent 1 will get the least valuable piece among $[c_1,c_2)$ and $[c_2,1]$ before the splitting of the trimmings. $T$ will be split so that the leftmost part is $P_1\cup P_2$ and the other two parts have equal value for agent 1; agent 1 will receive one of the two rightmost parts of $T$. In total, agent 1 will get at most $1/3$ (half of what $P_3\cup P_4\cup P_5\cup P_6$ is worth).

\item $c_1\in P_4,c_2\in P_6$. Assume first that $c_1$ is not the right boundary of $P_4$. It holds that $T=[r,c_1)$, where $r$ is the left boundary of $P_3$. Agent 1 will get piece $[c_2,1)$ before the splitting of the trimmings. Agent 1 will ony have positive value for one part of the trimmings and by our assumption she won't be allowed to take it. She cannot get utility more than $1/3$ overall. If $c_1$ is the right boundary of $P_4$, then the dominating strategy of agent 1 is the one we claim to be her best responce, i.e. where $c_2$ is the left boundary of $P_6$, which leads to utility $4/9$ for agent 1.  

\item $c_1,c_2\in P_4$. $T=[r,c_1)$, where $r$ is the left boundary of $P_3$. Agent 1 will get piece $[c_1,c_2)$ before the splitting of the trimmings, and she cannot get utility more than $1/3$ overall.

\item $c_1\in P_2,c_2\in P_6$. Either $T=[c_1,r)$ for $r \in P_5$, or $T'=[r',c_2)$ for $r' \in P_3$, depending on which has less utility for agent 1 by our assumption. Agent 1 will get part $[c_2,1]$ before the splitting of the trimmings. In either case, $T$ will be split in three parts, one of which will contain part $P_4$, and agent 1 won't be allowed to take that part. Overall, agent 1 will not get utility more than $1/3$ in either case.

\item $c_1\in P_2,c_2\in P_4$. $T=[c_1,r)$, where $r$ is the right boundary of $P_2$.
Agent 1 will get part $[0,c_1]$ before the splitting of the trimmings, and her overall utility will never be more than $1/3$.

\item $c_1,c_2\in P_2$. $T=[c_2,r)$, where $r\in 5$ and agent 1 will get part $[c_1,c_2]$ before the splitting of the trimmings. $T$ will be split in three parts, one of which will contain part $P_4$. By our assumption, agent 1 will not get that part of the trimmings, so overall she will have value at most $1/3$.

\end{itemize}
The proof is now complete.
\end{proof}
\end{appendices}
 
\end{document}